\documentclass[12pt]{article}
\RequirePackage{calc}
\RequirePackage[american]{babel}

\usepackage{enumerate,amsthm}
\RequirePackage{fancyhdr}
\newtheorem{theorem}{Theorem}[section]
\newtheorem{lemma}[theorem]{Lemma}

\RequirePackage{amsmath}
\usepackage{graphicx}
\usepackage{subfigure}
\usepackage{bbm}
\usepackage{xr}
\usepackage{booktabs}
\usepackage{multirow}
\usepackage[]{color}

\definecolor{fgcolor}{rgb}{0.345, 0.345, 0.345}

\definecolor{shadecolor}{rgb}{.97, .97, .97}
\definecolor{messagecolor}{rgb}{0, 0, 0}
\definecolor{warningcolor}{rgb}{1, 0, 1}
\definecolor{errorcolor}{rgb}{1, 0, 0}

\usepackage{alltt}

\usepackage{natbib}
\usepackage{amsfonts}
\usepackage{hyperref}
\RequirePackage{tikz}
\usetikzlibrary{fit,positioning,arrows,calc}
\usepackage{subfloat}
\def\bSig\mathbf{\Sigma}

\newcommand{\ide}{\mathbf{1}}

\newcommand{\eenum}{\end{enumerate}}

\newcommand{\bA}{\mathbf{A}}
\newcommand{\bh}{\mathbf{h}}
\newcommand{\bH}{\mathbf{H}}

\newcommand{\bL}{\mathbf{L}}
\newcommand{\bF}{\mathbf{F}}
\newcommand{\bY}{\mathbf{Y}}
\newcommand{\bx}{\mathbf{x}}

\newcommand{\bI}{\mathbf{I}}
\newcommand{\bg}{\mathbf{g}}
\newcommand{\df}{\mathrm{d}}
\newcommand{\bZ}{\mathbf{Z}}
\newcommand{\teta}{\tilde{\boldsymbol{\eta}}}

\newcommand{\bX}{\mathbf{X}}
\newcommand{\bW}{\mathbf{W}}
\newcommand{\bxi}{\boldsymbol{\xi}}
\newcommand{\bdelta}{\boldsymbol{\delta}}
\newcommand{\bbeta}{\boldsymbol{\beta}}
\newcommand{\bM}{\mathbf{M}}

\newcommand{\cI}{\mathcal{I}}
\newcommand{\btX}{\tilde{\mathbf{X}}}
\newcommand{\bttheta}{\tilde{\boldsymbol{\theta}}}

\newcommand{\bmu}{\boldsymbol{\mu}}
\newcommand{\bP}{\mathbf{P}}
\newcommand{\bbm}{\mathbf{m}}
\newcommand{\bSigma}{\boldsymbol{\Sigma}}

\newcommand{\bE}{\mathbf{E}}
\newcommand{\bU}{\mathbf{U}}
\newcommand{\Var}{\mathbf{Var}}

\newcommand{\btheta}{\boldsymbol{\theta}}

\newcommand{\etab}{\boldsymbol{\eta}}
\newcommand{\bPhi}{\boldsymbol{\Phi}}
\usepackage[figuresright]{rotating}
\usepackage{blkarray}
\usepackage{ulem}
\newcommand\redsout{\bgroup\markoverwith{\textcolor{red}{\rule[0.5ex]{2pt}{1pt}}}\ULon}

\usepackage{algorithm}
\usepackage[noend]{algpseudocode}

\begin{document}


\label{firstpage}

\begin{center}
	\Large{Fitting stochastic epidemic models to gene genealogies using linear noise approximation}\\\ \\
	
	\normalsize
	Mingwei Tang$^1$, Gytis Dudas$^{2,3}$, Trevor Bedford$^2$, Vladimir N. Minin$^{4,*}$
	
	\small
	$^1$Department of Statistics, University of Washington, Seattle\\
	$^2$Vaccine and Infectious Disease Division, Fred Hutchinson Cancer Research Center\\
	$^3$Gothenburg Global Biodiversity Centre (GGBC), Gothenburg, Sweden\\
	$^4$Department of Statistics, University of California, Irvine\\
	$^*$Corresponding author; \url{vminin@uci.edu}
	
	\normalsize
\end{center}


\begin{abstract}
Phylodynamics is a set of population genetics tools that aim at reconstructing demographic history of a population based on molecular sequences of individuals sampled from the population of interest. One important task in phylodynamics is to estimate changes in (effective) population size. 
When applied to infectious disease sequences such estimation of population size trajectories can provide information about changes in the number of infections. 
To model changes in the number of infected individuals, current phylodynamic methods use non-parametric approaches, parametric approaches, and stochastic modeling in conjunction with likelihood-free Bayesian methods. 
The first class of methods yields results that are hard-to-interpret epidemiologically. 
The second class of methods provides estimates of important epidemiological parameters, such as infection and removal/recovery rates, but ignores variation in the dynamics of infectious disease spread. 
The third class of methods is the most advantageous statistically, but relies on computationally intensive particle filtering techniques that limits its applications.
We propose a Bayesian model that combines phylodynamic inference and stochastic epidemic models, and
achieves computational tractability by using a linear noise approximation (LNA) --- a technique that allows us to approximate probability densities of  stochastic epidemic model trajectories. 
LNA opens the door for using modern Markov chain Monte Carlo tools to approximate the joint posterior distribution of the disease transmission parameters and of high dimensional vectors describing unobserved changes in the stochastic epidemic model compartment sizes (e.g., numbers of infectious and susceptible individuals). 
We apply our estimation technique to Ebola genealogies estimated using viral genetic data from the 2014 epidemic in Sierra Leone and Liberia. \ \\
\end{abstract}
%
%

{\bf Key words:} Coalescent, Susceptible-Infectious-Recovered model, state-space model, phylodynamics, Ebola virus
\section{Introduction}
\label{s:intro}
Phylodynamics is an area at the intersection of phylogenetics and population genetics that studies how epidemiological, immunological, and evolutionary processes affect viral phylogenies constructed based on molecular sequences sampled from the population of interest \citep{grenfell2004unifying, volz2013viral}. 
Phylodynamics is especially useful in infectious disease modeling because genetic data provide a source of information that is complimentary to the traditional  disease case count data. 
Here we are interested in inferring parameters governing infectious disease dynamics from the genealogy/phylogeny estimated from infectious disease agent molecular sequences collected during the disease outbreak. Working in a Bayesian framework, we develop an efficient Markov chain Monte Carlo (MCMC) algorithm that allows us to work with stochastic models of infectious disease dynamics, properly accounting for stochastic nature of the dynamics.
\par

Currently, learning about population-level infectious disease dynamics from molecular sequences can be accomplished using three general strategies. The first strategy relies on the coalescent theory --- a set of population genetics tools that specify probability models for genealogies relating individuals randomly sampled from the population of interest \citep{kingman1982coalescent, griffiths1994ancestral, donnelly1995coalescents}. Using a subset of these models  \citep{griffiths1994}, it is possible to estimate changes in effective population size --- the number of breeding individuals in an idealized population that evolves according to a  Wright-Fisher model \citep{wright1931evolution}. 
Such reconstruction can be done assuming parametric \citep{kuhner1998maximum, drummond2002estimating} or nonparametric \citep{drummond2002estimating, drummond2005bayesian, minin2008smooth, palacios2013gaussian, gill2013} functional forms of the effective population size trajectory.
In the context of infectious disease phylodynamics, nonparametric inference is the norm and the estimated effective population size is often interpreted as the effective number of infections or the effective number of infectious individuals. 
However,  reconstructed effective population size trajectories are not easy to interpret and 
estimation of parameters of disease dynamics is difficult to accomplish if one wishes to maintain statistical rigor \citep{pybus2001epidemic, frost2010viral}.

\par
Another way to learn about infectious disease dynamics from molecular sequences is to model explicitly events that occur during the infectious disease spread and to link these events to the genealogy/phylogeny of sampled individuals using birth-death processes. 
For example, a Susceptible-Infectious-Removed (SIR) model includes two possible events: infections and removals (e.g., recoveries and deaths), represented by births and deaths in the corresponding birth-death model \citep{stadler2013birth,kuhnert2014simultaneous}. 
Other SIR-like models (e.g., SI and SIS models) differ by the number and types of the events that are needed to accurately describe natural history of the infectious disease \citep{leventhal2013using}. 
Although these methods are more principled than post-hoc processing of nonparametrically estimated disease dynamics, they are not easy to scale to large datasets and/or high dimensional models. 
For example, in order to fit phylodynamic birth-death models to genomic and epidemiological data \citet{vaughan2018directly} use particle filter MCMC.
However, computational burden of particle filter MCMC methods is usually very high. 
Moreover, these methods often struggle with convergence when the dimensionality of statistical model parameters is even moderately high \citep{andrieu2010particle}.

Structured coalescent models provide the third strategy of inferring parameters governing spread of an infectious disease \citep{volz2009phylodynamics,volz2012complex, dearlove2013coalescent}. These models assume infectious disease agent genetic data have been obtained from a random sample of infected individuals, allowing for serial sampling over time. Although similar to the birth-death modeling framework, the structured coalescent models have two advantages. First, one does not have to keep track, analytically or computationally, of extinct and not sampled genetic lineages. Second, the density of the genealogy can be obtained given the population level information about status of individuals: for example, in the SIR model it is sufficient to know the numbers of susceptible, ($S(t)$), infectious, ($I(t)$), and recovered, ($R(t)$), individuals at each time point $t$. The second advantage comes with two caveats: 1) such densities can be obtained only approximately and 2) evaluating densities of genealogies is not straightforward and involves numerical solutions of differential equations.
Even in cases when these caveats are manageable, the density of the assumed stochastic epidemic model population trajectory remains computationally intractable. 
One way around this intractability assumes a deterministic model of infectious disease dynamics \citep{volz2009phylodynamics,volz2012complex,volz2014phylodynamic}, which potentially leads to overconfidence in estimation of model parameters. 
Particle filter MCMC offers another solution \citep{rasmussen2011inference, rasmussen2014phylodynamic}, but, as we discussed already, these methods are difficult to use in practice, especially in high dimensional parameter spaces.
\par
In this paper, we develop methods that allow us to bypass computationally unwieldy particle filter MCMC with the help of a linear noise approximation (LNA).
LNA is a low order correction of the deterministic ordinary differential equation describing the asymptotic mean trajectories of compartmental models of population dynamics defined as Markov jump processes (e.g., chemical reaction models and SIR-like models of infectious disease dynamics) \citep{kurtz1970solutions,kurtz1971limit, van1983stochastic}. 
LNA can also be viewed as a first order Taylor approximation of Markov population dynamics models represented by stochastic differential equations \citep{giagos2010inference,wallace2010simplified}. 
A key feature of the LNA method is that it approximates the transition density of a stochastic population model with a Gaussian density \citep{komorowski2009bayesian}. 
\par
Inspired by recent applications of LNA to analysis of Google Flu Trends data \citep{fearnhead2014inference} and disease case counts \citep{buckingham2018gaussian}, we develop a Bayesian framework that combines LNA for stochastic models of infectious disease dynamics with structured coalescent models for genealogies of infectious disease agent genetic samples.
Our approach yields a latent Gaussian Markov model that closely resembles a Gaussian state-space model.
We use this resemblance to develop an efficient MCMC algorithm that combines high dimensional elliptical slice sampler updates \citep{murray2010elliptical} with low dimensional Metropolis-Hastings (MH) moves.
Using simulations, we demonstrate that this algorithm can handle reasonably complex models, including an SIR model with a time-varying infection rate. 
We apply this SIR model to a recent Ebola outbreak in West Africa.
Our analysis of data from Liberia and Sierra Leone illuminates significant changes in the Ebola infection rate over time, likely caused by the public health response measures and increased awareness of the outbreak in the population.

 \section[]{Methodology}
\label{s:method}
\subsection{Genealogy as data}
\label{s:coallike}
We start with $n$ infectious disease agent molecular sequences obtained from infected individuals sampled uniformly at random from the total infected population.
Further, we assume that a phylogenetic tree, or genealogy, $\mathbf{g}$ relating these sequences has been estimated in such a way that the tree branch lengths respect the known sequence sampling times.
Such estimation can be performed with, for example, \texttt{BEAST}  --- a leading software package for Bayesian phylogenetic studies, particularly popular among molecular epidemiologists who collect and analyze viral genetic sequences \citep{suchard2018bayesian}. 
The genealogy is represented by a tree structure with its nodes containing two sources of temporal information: coalescent and sampling times. 
The coalescent times correspond to the internal nodes of the tree, which are defined as the times at which two lineages in the tree are merged into a common ancestor. 
The sampling times, corresponding to the tips of the tree, are the times at which molecular sequences were sampled.
Note that sampling times are observed directly, while coalescent times are estimated from molecular sequences during phylogenetic reconstruction. 
\par
To perform inference about infectious disease dynamics using the above genealogy we need a probability model that relates the genealogy and infectious disease dynamics model parameters. 
Without too much loss of generality, we assume that the infectious disease is spreading through the population according to the SIR model --- a compartmental model that at each time point $t$ tracks the number of susceptible individuals $S(t)$, number of infected/infectious individuals $I(t)$, and number of removed individuals $R(t)$ \citep{bailey1975mathematical, anderson1992infectious}. 
We assume that the population is closed so $S(t) + I(t) + R(t) = N$ for all times $t$, where $N$ is the population size that we assume to be known. 
This constraint implies that vector $\bX(t) = (S(t), I(t))$ is sufficient to keep track of the population state at time $t$. 
We follow common practice and model $\bX(t)$ as a Markov jump process (MJP) with allowable 
instantaneous jumps shown in Figure~\ref{f:ctmc} \citep{ONeill1999}. 
The assumed MJP process $\bX(t)$ is inhomogeneous, because we allow the infection rate $\beta(t)$ and removal rate $\gamma(t)$ to be time-varying. 
\begin{figure}
	\centering
	\begin{tikzpicture}
	\usetikzlibrary{shapes}
	\tikzstyle{main}=[ellipse, minimum size = 10mm, thick, draw =black!80, node distance = 16mm]
	\tikzstyle{connect}=[-latex, thick]
	\tikzstyle{box}=[rectangle, draw=black!100]
	\node[main] (theta)  {$S,I-1,R+1$};
	\node[main] (z) [right=of theta] {$S,I,R$};
	\node[main] (w) [right=of z] { $S-1, I+1, R$};
	\path (z) edge [ connect] node [midway,above=0.1em]  {$\gamma(t) I$} node[midway,below=0.1em]{$\mbox{removal}$}(theta)
	(z) edge [connect] node [midway,above=0.1em]  {$\beta(t) SI$}node[midway,below=0.1em]{$\mbox{infection}$}(w);
	\end{tikzpicture}
	\caption{SIR Markov jump process. From the current state with the counts $S,I,R$, the population can transition to state $S-1, I+1, R$ (an infection event) with rate $\beta(t) SI$ or  to state $S,I-1,R+1$ (a removal event) with rate $\gamma(t) I$. No other instantaneous transitions are allowed.}
	\label{f:ctmc}
\end{figure}
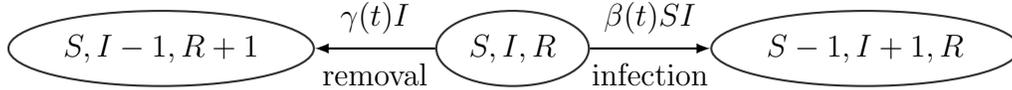
\par 
The structured coalescent models assume that only coalescent times $c_1<c_2<\cdots<c_{n-1}$  provide information about the population dynamics. 
These times are modeled as jumps of an inhomogeneous pure death process with rate $\lambda(t)$, where each ``death" event corresponds to coalescence of two lineages and $\lambda(t)$ is called a coalescent rate. 
Then the density of the genealogy, which serves as a likelihood in our work, is written as 
\[
\Pr (\mathbf{g}) \propto \prod_{k=2}^{n} \lambda(c_{k-1}) \exp\left(-\int_{c_{k-1}}^{c_{k}} \lambda(\tau) \mathrm{d}\tau\right),
\]
where $c_n$ denotes the most recent sequence sampling time.
The dependence of coalescent rate on the assumed population dynamics can be complicated and mathematically intractable, but luckily  approximations exist for some specific cases.   
For the SIR model the approximate coalescent rate can be obtained via the following formula:
\begin{eqnarray}
\label{eq:volz}
\lambda(t) = \lambda(l(t), \beta(t), \mathbf{X}(t)) = 
 \binom{l(t)}{2}\dfrac{2 \beta(t) S(t)}{I(t)}, 
\end{eqnarray}
where $l(t)$ is the number of  lineages present at time $t$. 
Note that when the number of susceptibles is not changing significantly relative to the total population size (i.e., $S(t) \approx N$) and infection rate is constant (i.e., $\beta(t) = \beta$), the structured coalescent reduces to the classical Kingman's coalescent, where we interpret $I(t)/(2 \beta N)$ as the effective population size trajectory \citep{kingman1982coalescent}.
It is possible to find approximate coalescence rate for general compartmental models, but closed form expressions exist only for a few models with a low number of compartments (e.g., SI, SIR) \citep{volz2009phylodynamics,volz2012complex, dearlove2013coalescent}.

Since we allow sequences to be sampled at different times $s_1< s_2< \cdots < s_m=c_n$, some inter-coalescent times are censored. 
To deal with this censoring algebraically,  each inter-coalesecent interval $[c_{k-1},c_k)$ is partitioned by the sampling events into $i_k$ sub-intervals: $\cI_{0,k},\ldots,\cI_{i_k-1,k}$. The intervals that end with a coalescent event are defined as
$\cI_{0,k} = [c_{k-1},\min\{c_k,s_j\})$, for $s_j>c_{k-1}$ and $k=2,\ldots,n$. 
Let the number of lineages in each interval $\cI_{i,k}$ be $l_{i,k}$.
Then the number of lineages at each time point $t$ can be written as $l(t) = \sum_{k=2}^n \sum_{i=0}^{i_{k-1}} 1_{\{t \in I_{i,k}\}} l_{i,k}$.
If the interval $\cI_{i,k}$ ends with a coalescent time, the number of lineages in the next interval will be decreased by $1$. 
If the interval ends with a sampling event $s_i$, then the number of lineages in the next interval is increased by $n_i$ --- the number of sequences sampled at time $s_i$. Figure \ref{f:tree} shows an example of a genealogy with labeled coalescent times, sampling times, number of lineages, and the corresponding intervals. 
\par
We are now ready to connect the SIR model and a genealogy with serially sampled tips with the help of a structured coalescent density/likelihood. 
First we discretize the time interval between the time to most recent common ancestor $c_1$ (time corresponding to the root of the tree) and the most recent sampling time $s_m$ using a
regular grid $t_0<t_1<\cdots<t_T$ ($t_0 < c_1$ and $t_T> s_m$).
Using this grid, we discretize the latent epidemic trajectory by assuming that
$\bX(t) = \sum_{j=1}^{T}\bX_{j-1}\mathbf{1}_{[t_{j-1},t_j)}(t)$, where $\bX_{j} = (S_{j}, I_{j})$ is a column vector. 
Similarly, we discretize the infectious disease dynamics parameter vector trajectory $\btheta(t) = (\beta(t), \gamma(t))$ so that $\btheta(t) = \sum_{j=1}^{T} \btheta_{j-1} \mathbf{1}_{[t_{j-1},t_j)}(t)$, where $\btheta_{j} = (\beta_j, \gamma_j)$ is also a column vector.
We collect latent variables $\bX_j$s and parameters $\btheta_j$s into matrices $\bX_{0:T}$ and $\btheta_{0:T}$ respectively. 
The SIR structured coalescent density/likelihood then becomes
\begin{eqnarray}
\label{eq:coalgeneral}
\Pr (\mathbf{g} \mid \bX_{0:T},\btheta_{0:T}) \propto \prod_{k=2}^{n} \binom{l(c_{k-1})}{2}\dfrac{2 \beta(c_{k-1}) S(c_{k-1})}{I(c_{k-1})}\exp\left(-\sum_{i=0}^{i_k-1}\int_{\cI_{i,k}} \binom{l_{i,k}}{2}   \dfrac{2 \beta(\tau) S(\tau)}{I(\tau)} \df \tau\right).
\end{eqnarray}
\label{s:coal}
Since $S(t)$, $I(t)$, and $\beta(t)$ are piecewise constant functions, the integrals in the above formula are readily available in closed form and are fast to compute.

\begin{figure}
	\label{f:tree}
	\centerline{\includegraphics[width=6in]{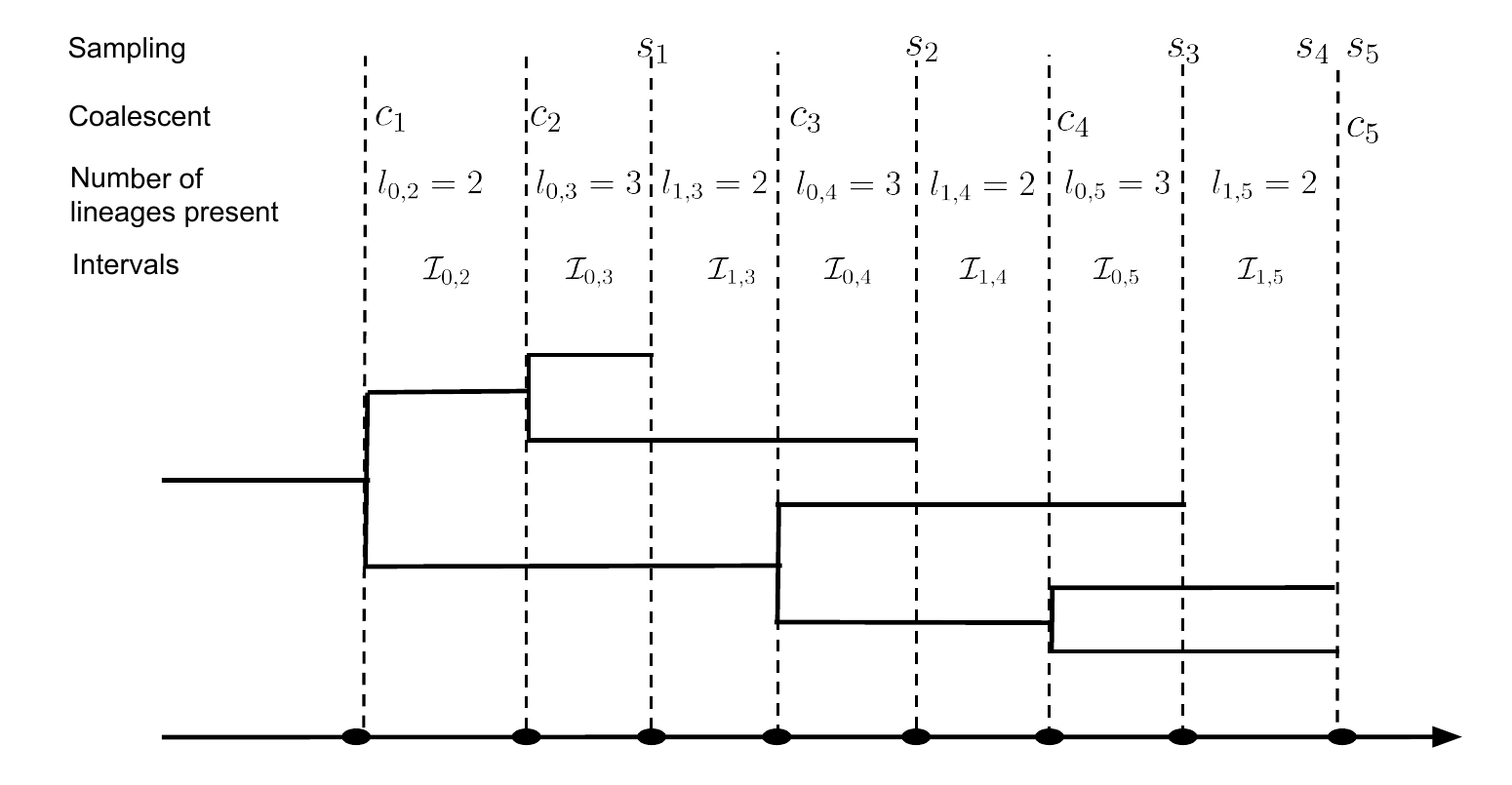}}
	\caption{Example of a genealogy. Black solid lines show the genealogy structure. The colescent times $c_1,\ldots,c_4$ and sampling times $s_1,\ldots,s_4$ are labeled with vertical dashed lines. The number of lineages $l_{i,k}$ is given in each intervals $\cI_{i,k}$.}
\end{figure}


\subsection{Bayesian data augmentation}
\subsubsection{Posterior distribution}
\label{s:SIR}
 Given genealogy $\bg$, our goal is to infer the latent SIR population dynamic $\bX_{0:T}$ and rate parameters $\btheta_{0:T}$ over time grid $t_0<t_1<\cdots<t_T$. Let $\Pr(\bX_{0})$ and $\Pr(\btheta_{0:T})$ denote the prior densities for the initial compartment states and the SIR parameters respectively. The posterior distribution for the population trajectory $\bX_{0:T}$ and parameters 
$\btheta_{0:T}$ given observed genealogy $\bg$ is 
\begin{eqnarray}
\label{eq:bayes}
\Pr\left(\bX_{0:T},\btheta_{0:T} \mid \bg\right) \propto \Pr\left(\bg\mid\bX_{0:T}, \btheta_{0:T}\right)\Pr\left(\bX_{1:T}\mid\bX_{0},\btheta_{0:T}\right)\Pr\left(\btheta_{0:T}\right)\Pr\left(\bX_0\right), 
\end{eqnarray}
where $\Pr\left(\bg\mid\bX_{0:T}, \btheta_{0:T}\right)$ is the structured coalescent likelihood introduced in Section \ref{s:coal} and $\Pr(\bX_{1:T}\mid \bX_0,\btheta_{0:T})$ is the likelihood function for discrete observations of trajectory $\mathbf{X}_{1:T}$ given the initial value $\bX_0$:
\begin{eqnarray}
\label{eq:factorize}
\Pr(\bX_{1:T}\mid\bX_0,\btheta_{0:T}) = \prod_{i=1}^{T}\Pr(\bX_i\mid\bX_{i-1}, \boldsymbol{\theta}_{i-1}),
\end{eqnarray}
where the factorization comes from the assumed Markov property of the disease dynamics.
However, the SIR transition density $\Pr(\bX_i\mid\bX_{i-1}, \btheta_{i-1})$ becomes intractable as population size $N$ grows large, making it difficult to perform likelihood-based inference for outbreaks in large populations.

\subsubsection{Linear noise approximation}
\label{s:LNA}
To furnish a feasible computation strategy for large populations, we use a linear noise approximation (LNA) method, in which the computationally intractable transition probability $\Pr\left(\bX_i\mid\bX_{i-1}, \btheta_{i-1}\right)$ is approximated using a closed form Gaussian transition density. 
The LNA method replaces the MJP discrete state space with a continuous state space of $\bX(t)$ to approximate the counts of at time $t$, under the following constraints: $S(t) > 0, I(t) > 0$ and $S(t) + I(t)\leq N$. 
To briefly explain how this approximation is obtained, we will need additional notation.
\par
The SIR MJP instantaneous transitions, depicted in Figure~\ref{f:ctmc}, are encoded in an effect matrix 
\begin{equation}
\label{eq:effectmx}
\bA = \begin{blockarray}{ccl}
\mbox{susceptible} & \mbox{infected} \\
\begin{block}{(cc)c}
-1 & 1 & \mbox{infection}\\
0 & -1 & \mbox{removal}.\\
\end{block}
\end{blockarray}
\end{equation}
Each row in matrix \eqref{eq:effectmx} represents a type of transition event and each column corresponds to a change in the susceptible and infected populations. 
Next, we define a rate vector $\bh$ and a rate matrix $\bH$:
\begin{eqnarray}
\label{eq:hH}
\bh(\bX(t),\btheta(t)) = \left(\begin{array}{c}
\beta(t) S(t)I(t)\\
\gamma(t) I(t)
\end{array}\right), \bH = \mbox{diag}\left(\bh(\bX(t),\btheta(t))\right)=\left(
\begin{array}{cc}
\beta(t) S(t)I(t) & 0\\
0 & \gamma(t) I(t)
\end{array}
\right).
\end{eqnarray}
The above notation, as well as subsequent developments based on it, can be generalized
to other epidemic models and, more generally, to a large class of density dependent stochastic processes, 
such as chemical reaction and gene regulation models  \citep{wilkinson2011stochastic}. 
See  Section \ref{s:generalization} in the Appendix for more details on this generalization.

Consider a transition from $\bX_{i-1}$ at time $t_{i-1}$ to $\bX_{i}$ at $t_i$.
Recall that we assume that the SIR rates $\btheta(t)$ take constant values $\btheta_{i-1}$ in $[t_{i-1}, t_i)$. 
The LNA represents the value of the next state $\bX_{i}$ as $\bX_i = \etab(t_i) + \bM(t_i)$, where $\etab(t_i)$ is a 
deterministic component and $\bM(t_i)$ is a stochastic component.
The deterministic component $\etab(t_i)$  can be obtained by solving the standard SIR ODE that in our notation can be written as
\begin{eqnarray}
\label{eq:ODEeta}
\df \etab(t) = \bA^T \bh(\etab(t), \btheta_{i-1}) \df t,  \quad t \in [t_{i-1}, t_i].
\end{eqnarray}
The stochastic part $\bM(t_i)$ corresponds to the solution of the following SDE at time $t_i$:
\begin{eqnarray}
\label{eq:bM}
\df \bM(t) = \bF(\etab(t),\btheta_{i-1}) \bM(t) \df t + \sqrt{\bA^T\bH(\etab(t),\btheta_{i-1}) \bA }\df \bW_t, \quad t \in [t_{i-1}, t_i], 
\end{eqnarray}
where $\bF(\etab(t),\btheta_{i-1}):= \dfrac{\partial \bA^T\bh(\bX(t),\btheta_{i-1})}{\partial \bX} \Big |_{\bX = \etab(t)}$ is the Jacobian matrix of the deterministic part $\bA^T\bh(\bX(t),\btheta_{i-1})$ in \eqref{eq:ODEeta} evaluated at $\etab(t)$. 
The solution of SDE \eqref{eq:bM}, $\bM(t)$, is a Gaussian process and can be recovered by solving two ordinary differential equations governing the mean function $\bbm(t):=\bE [\bM(t)]$ and covariance function $\bPhi(t) := \Var(\bM(t))$:
\begin{eqnarray}
\label{eq:m}\df \bbm(t) &=&  \bF(\etab(t),\btheta_{i-1}) \bbm(t)\df t,\\
\label{eq:Phi}\df \bPhi(t) &=& \left(\bF\left(\etab(t),\btheta_{i-1}\right) \bPhi(t) + \bPhi(t)\bF^T(\etab(t),\btheta_{i-1}) + \bA^T\bH\left(\etab(t),\btheta_{i-1}\right)\bA \right)\df t, 
\end{eqnarray}
for $t\in [t_{i-1}, t_i]$. A heuristic derivation of LNA, based on \citet{wallace2010simplified}, is given in Section \ref{s:LNAderive} of the Appendix. 
Let $\etab_{t_{i-1}}, \bbm_{t_{i-1}}, \bPhi_{t_{i-1}}$ denote the initial values of $\etab(t), \bbm(t), \bPhi(t)$ at time $t_{i-1}$ in differential equations \eqref{eq:ODEeta}, \eqref{eq:m}, and \eqref{eq:Phi} respectively. 
There are two options for choosing these initial conditions: the non-restarting LNA of \citet{komorowski2009bayesian} and the restarting LNA of \citet{fearnhead2014inference}.  
In this paper, we will use the non-restarting LNA by \citet{komorowski2009bayesian} with the following choice of initial conditions:
\begin{enumerate}
	\item $\etab_{t_{i-1}} = \etab(t_{i-1})$, where $\etab(t_{i-1})$ was obtained by solving the ODE \eqref{eq:ODEeta} using parameter vector $\btheta_{i-2}$ over the interval $[t_{i-2}, t_{i-1}]$, 
	\item $\bbm_{t_{i-1}} =\bX_{i-1} - \etab(t_{i-1})$, 
	\item $\bPhi_{t_{i-1}} = \mathbf{0}$. 
\end{enumerate}
Solving the system of ODEs \eqref{eq:ODEeta}, \eqref{eq:m}, \eqref{eq:Phi}, we obtain $\etab(t_i)$, $\bbm(t_i)$, and $\bPhi(t_i)$. The solution $\bbm(t_i)$ will be a function of the initial value $\bX_{i-1} - \etab(t_{i-1})$, the interval length $\Delta t_i := t_{i} - t_{i-1}$ and the SIR rates $\btheta_{i-1}$. To make this dependence explicit, we write $\bbm(t_i):=\bmu\left(\bX_{i-1} - \etab(t_{i-1}),\Delta t_i, \btheta_{i-1}\right)$. Since \eqref{eq:m}  is a first order homogeneous linear ODE, the solution $\bmu\left(\bX_{i-1} - \etab(t_{i-1}),\Delta t_i, \btheta_{i-1}\right)$ is a linear function of $\bX_{i-1} - \etab(t_{i-1})$.
Hence, the transition from $\bX_{i-1}$ to $\bX_{i}$ follows the following Gaussian distribution:
\begin{eqnarray}
\label{eq:transition}
\bX_i \mid \bX_{i-1},\btheta_{i-1} \sim \mathcal{N}\left(\etab(t_i) + \bmu\left(\bX_{i-1} - \etab(t_{i-1}),\Delta t_i, \btheta_{i-1}\right),\bPhi(t_i) \right).
\end{eqnarray}
To summarize, the derived conditional Gaussian densities $\Pr(\bX_i \mid \bX_{i-1},\btheta_{i-1})$ allow us to  compute the density of the latent SIR trajectory \eqref{eq:factorize}. As a result, our augmented posterior distribution of $\bX_{0:T}$ and $\btheta_{0:T}$, shown in equation \eqref{eq:bayes}, can be computed up to proportionality constant and approximated via ``standard'' (not particle filter) MCMC approaches.

\subsection{Reparameterization, priors, and MCMC algorithm}
\label{s:inference}
\subsubsection{Reparameterizing SIR rates}
\label{s:sir_parameterization}
We have experimented with multiple parameterizations of our inhomogeneous SIR model and found that the following parameterization works best with our proposed MCMC algorithm for approximating the posterior distribution \eqref{eq:bayes}.
First, recall that we allow SIR rates to vary with time. 
Since it is much more likely for the infection rate to be time variable, we are going to assume a constant removal/recovery rate $\gamma$.
This leaves us with the following parameters: infection rates on a grid $\boldsymbol{\beta}$, removal rate $\gamma$, and initial SIR state $\mathbf{X}_0 = (S_0, I_0)$.
Since we are interested in modeling an emerging infectious disease outbreak, we set the initial counts of susceptibles to $S_0 = N - I_0$.
Initial counts of infected individuals, $I_0$, is assumed to be low and treated as an unknown parameter with a lognormal prior distribution.
Instead of the time-varying infection rate $\beta(t)$, we parameterize our SIR model with a time-varying basic reproduction number $R_{0}(t) = [\beta(t) N]/ \gamma$.
The reproduction number is interpreted as the average number of cases that one case generates over its infectious period in a completely susceptible population. 
Since our infection rate changes in a piecewise manner, the basic reproduction number varies over time in a piecewise manner too: 
	\begin{eqnarray}
	\label{eq:changeOriginal}
	R_0(t) = \sum_{i=1}^{T} R_{0_{i-1}}\ide_{[t_{i-1},t_i)}(t), 
	\end{eqnarray} 
  where $R_{0_i} = [\beta_i N]/\gamma$ is the reproduction number corresponding to the time interval $[t_{i-1}, t_i)$.
	Let $R_0 = R_{0_0}$ be the initial basic reproductive number and $\delta_i = \log\left(R_{0_{i}}/R_{0_{i-1}}\right)/\sigma$ be a normalized log ratio of $R_0(t)$ over two successive time intervals. 
	Then, interval-specific basic reproduction numbers can be written as
	\begin{eqnarray}
		R_{0_i}  = R_0(t,\bdelta_{1:T},\sigma)= R_{0} \exp\left(\sum_{k=1}^{i}\sigma \delta_{k}\right), \mbox{for } i = 1,\ldots, T,
	\end{eqnarray}	
	where we assume \textit{a priori} that $\delta_i$s are independent standard normal random variables.

	This construction implies that log-transformed piecewise constant reproduction numbers, $\log(R_{0_i})$s, \textit{a priori} follow a first order Gaussian Markov random field (GMRF) with standard deviation $\sigma$ that controls the \textit{a priori} smoothness of $R_0(t)$ trajectory \citep{rue2001fast,rue2005gaussian}.
	In addition to speeding MCMC convergence, working with  $R_0(t)$ is convenient, because this trajectory is dimensionless and retains its interpretation when one changes the population size $N$.
The initial $R_0$ is assigned a $\mbox{lognormal}(a_1,b_1)$ prior. 
We use a $\mbox{lognormal}(a_2,b_2)$ prior for the inverse of standard deviation $1/\sigma$.
	
	
\subsubsection{Reparameterizing SIR latent trajectories}	
	We reparameterize the latent SIR trajectory $\bX_{1:T}$ with a sequence of independent Gaussian random variables $\bxi_{1:T}$, following a non-centered parameterization framework of \citet{papaspiliopoulos2007general}. 
	According to formula \eqref{eq:transition},  conditional on $\bX_{i-1}$, $\bX_{i}$ can be written as
	\begin{eqnarray}
	\label{eq:linear}
	\bX_{i} = \etab(t_i) + \bmu(\bX_{i-1} - \etab(t_{i-1}), \Delta t_i,\btheta_{i-1}) + \bPhi_i^{1/2}\bxi_{i},
	\end{eqnarray}
	where $\bxi_i \overset{\text{iid}}{\sim} \mathcal{N}(0,\bI)$ for $i = 1,\ldots,T$ and $\bI$ is a $2\times 2$ identity matrix. 
In our parameterization, we will treat $\bxi_{1:T}$ as random latent variables and the SIR latent trajectory $\bX_{1:T}$ as a deterministic transformation of $\bxi_{1:T}$. 
More details about our non-centered parameterization of $\bX_{1:T}$ can be found in Section \ref{s:noncenter} of the Appendix.

\subsubsection{MCMC algorithm}
Using our new parameterization, we are now interested in the posterior distribution of 
the initial number of infected individuals, $I_0$, removal rate, $\gamma$, the initial basic reproduction number, $R_0$, standardized vectors, $\bdelta_{1:T}$ and $\bxi_{1:T}$, 
and GMRF standard deviation, $\sigma$:
\[
\begin{aligned}
\Pr(I_0,R_0,\gamma, \bdelta_{1:T}, \bxi_{1:T}, \sigma|\bg)  &\propto  \Pr(\bg|I_0,R_0,\gamma, \bdelta_{1:T}, \bxi_{1:T}, \sigma)\Pr(I_0)\Pr(R_0)\Pr(\gamma)\Pr(\bdelta_{1:T})\Pr(\bxi_{1:T})\Pr(\sigma)\\
&\propto \Pr(\bg|\bX_{0:T}, \btheta_{0:T})\Pr(I_0)\Pr(R_0)\Pr(\gamma)\Pr(\bdelta_{1:T})\Pr(\bxi_{1:T})\Pr(\sigma).
\end{aligned}
\]
The latent variables $\bX_{0:T}$ and parameter vector $\btheta_{0:T}$ are deterministic functions of new parameters $I_0$,  $\gamma$, $R_0$, $\bdelta_{1:T}$, $\bxi_{1:T}$, and $\sigma$.
We use the following MCMC with block updates to approximate this posterior distribution.
We update high dimensional vector $\bU = (\log(R_0), \bdelta_{1:T}, \log(\sigma))$ using the efficient elliptical slice sampler \citep{murray2010elliptical}. 
Vector $\bxi_{1:T}$ is updated the same way in a separate step.
Initial number of invected individuals $I_0$ and removal rate $\gamma$ are updated using univariate Metropolis steps.
The full procedure is described in Algorithm~\ref{MCMCstep}, which together with details of the elliptical slice sampler can be found in Section \ref{s:ESS} of the Appendix.
After MCMC is done, we report posterior summaries using natural parameterization.
For example, we report posterior medians  and 95\% Bayesian credible intervals (BCIs) of the piecewise latent reproduction number trajectory, $R_{0_i}$ for $i = 0,\ldots,T$, and latent trajectory $\bX_{0:T}$.

\subsubsection{Implementation}
Our R package called \texttt{LNAPhylodyn} provides an implementation of our MCMC algorithm. The package code is publicly available at \url{https://github.com/MingweiWilliamTang/LNAphyloDyn}. This repository also contains scripts that should allow one to reproduce key numerical results in this manuscript.
 
\section{Simulation experiments}
\label{s:exp}
\subsection{Simulations based on single genealogy realizations}
\label{s:oneReal}
In this section, we use simulated genealogies to assess performance of our LNA-based method and to compare it with an ODE-based method, where we replace equation \eqref{eq:linear} with its simplified version: $\bX_{i} = \etab(t_i)$. 
Under our assumption of a fixed and known genealogy and constant $R_0$, our ODE-based method closely resembles  previously developed methods by \citet{volz2009phylodynamics} and \citet{volz2018bayesian}. 
To compare ODE-based and LNA-based models in a Bayesian nonparametric setting, we equip the ODE model with the GMRF prior for time-varying $R_0(t)$, described in Section \ref{s:sir_parameterization}.  
We use the same MCMC algorithm for both LNA-based and ODE-based models, except we do not have a separate step to update latent vector $\bxi_{1:T}$ (equivalently, $\bX_{0:T}$) in the ODE-based inference. See Algorithm~\ref{ode-mcmc} in the Appendix for a more detailed description of the ODE-based MCMC.
\par
The simulation protocol consists of two steps. 
First, given the population size $N$ and pre-specified parameters $\gamma$, $I_0$, and $R_0(t)$, we simulate one realization of the SIR population trajectory based on the MJP using the Gillespie algorithm \citep{gillespie1977exact}. 
Next, we generate realistic lineage sampling times and simulate coalescent times from the distribution specified by density \eqref{eq:coalgeneral} using a thinning algorithm by \citet{palacios2013gaussian}. 


We test LNA-based and ODE-based methods under three ``true'' $R_0(t)$ trajectories over the time interval $[0,90]$: 
\begin{enumerate}
	\item Constant (CONST) $R_0(t)$. $R_0(t) = 2.2$ for $t\in [0,90]$. Recovery rate $\gamma = 0.2$. Initial counts of infected individuals $I_0 = 1$.  Total population size is $N=$ 100,000. 
	\item Stepwise decreasing (SD) $R_0(t)$. $R_0(t) = 2,t\in [0,30)$, $R_0(t) =1, t\in [30,60)$ and $R_0(t) = 0.6, t\in[60,90]$. Recovery rate $\gamma = 0.2$. Initial counts of infected individuals $I_0 = 1$. Total population size $N = $ 1,000,000. 
\item Non-monotonic (NM) $R_0(t)$. $R_0(t) = 1.4\times 1.015^{0.5t}, t \in [0,30]$, $R_0(t) = 1.750 \times 0.975^{t-30},t \in [30,80]$ and $R_0(t) = 0.4583, t\in [80,90]$. Recovery rate $\gamma = 0.3$. Initial counts of infected individuals $I_0 = 3$. Total population size $N =$ 1,000,000. 
\end{enumerate}

For all simulations, we use $\mbox{lognormal}(1,1)$ prior for $I_0$. The parameters of the lognormal priors for the initial $R_0$ and inverse standard deviation $1/\sigma$ are set to $a_1 = 0.7, b_1 = 0.5$ and $a_2=3, b_2=0.2$ respectively, in such a way that \textit{a priori} $R_0(t)$ trajectory stayed within a reasonable range of $[0,5]$ with 0.9 probability. 
We assign an informative prior for $\gamma$ in each simulation scenario, because prior information about this parameter is usually available: (1) CONST: $\gamma \sim \mbox{lognormal}(-1.7, 0.1)$, (2) SD: $\gamma \sim \mbox{lognormal}(-1.7,0.1)$, (3) NM: $\gamma \sim \mbox{lognormal}(-1.2,0.1)$. We set the grid size to $T=36$, with $t_i - t_{i-1} = 2.5$ for $i = 1,\ldots,36$.
For both LNA-based and ODE-based methods, we use 300,000 MCMC iterations. 
All MCMC chains appeared to converge (trace plots are shown in Section \ref{s:traceplotExp} of the Appendix). 
The effective sample sizes of all unknown quantities were above 100.

\begin{figure}
	\centerline{\includegraphics[width=6.5in]{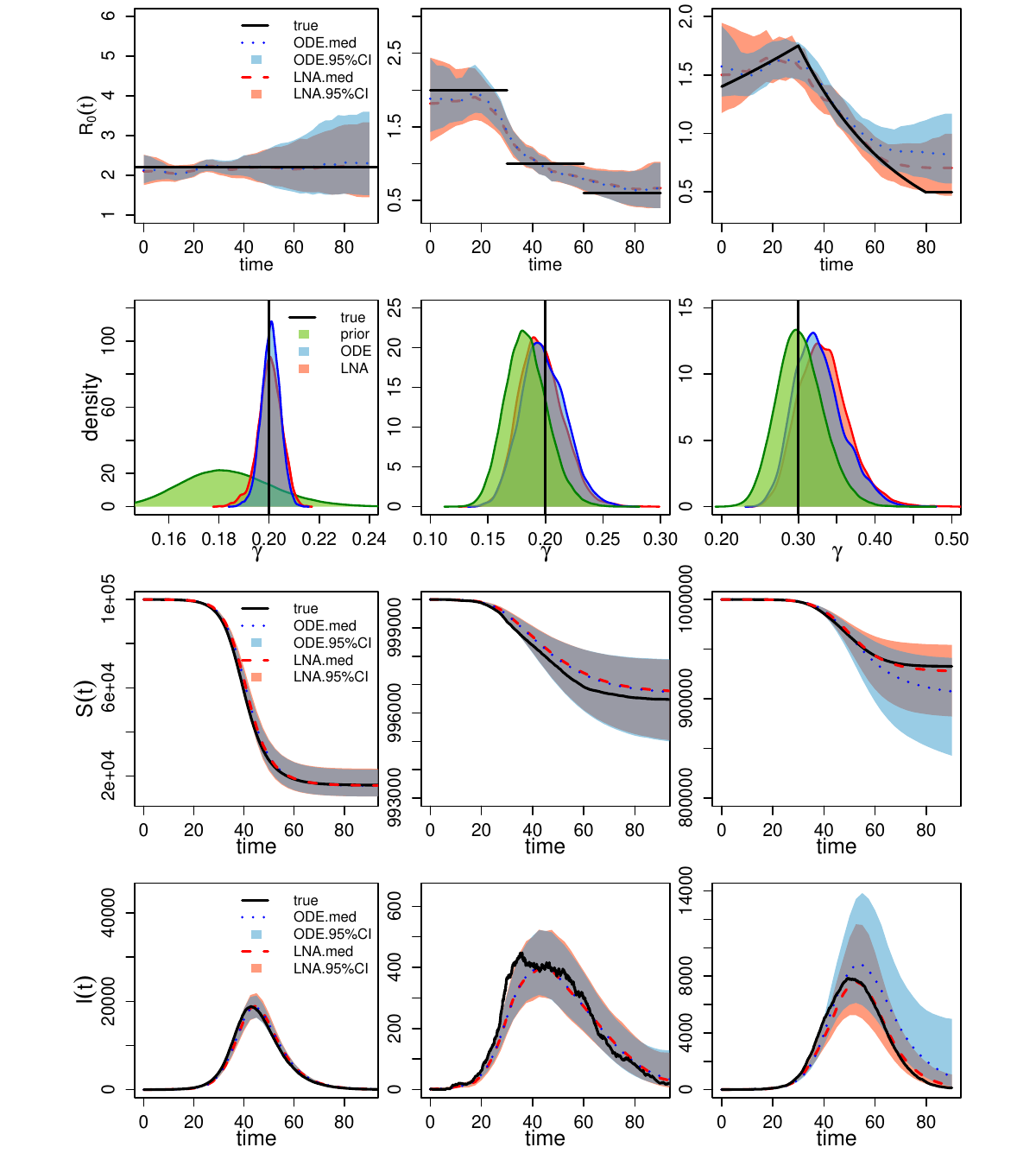}}
	\caption{Analysis of 3 simulation scenarios. Columns correspond to CONST, SD, and NM simulated $R_0(t)$ trajectories. 
	The first row shows the estimated $R_0(t)$ trajectories for the 3 scenarios, with the black solid lines representing the truth, the red dashed lines depicting the posterior median and the red-shaded area showing the 95\% BCIs for the LNA-based method. 
	For the ODE-based method, the posterior median is plotted in blue dotted lines, with blue shading showing the 95\% BCIs. The second row corresponds to the estimation for the removal rate $\gamma$. Posterior density curves from the LNA are shown in red lines and the posterior density for ODE is plotted in blue lines, compared with prior density curve in green lines. The bottom two figures shows the estimated trajectory of $S(t)$ and $I(t)$ respectively.
	}
	\label{f:simulation2}
\end{figure}
  
The first row of Figure~\ref{f:simulation2} shows point-wise posterior medians and 95\% BCIs for the basic reproduction number trajectory, $R_0(t)$.
Our LNA-based method performs well in capturing the continuous dynamics of $R_0(t)$. 
Though our approach may not perfectly catch the discontinuous changes in $R_0$ in the SD scenario, the method provides BCIs that are able to capture most of the $R_0(t)$ trajectory. 
The ODE-based method yields similar results in the CONST case and the SD case, but fails to capture the decreasing trends in the NM scenario. 

The second row in Figure \ref{f:simulation2} shows posterior summaries of removal rate $\gamma$. 
Both LNA-based and ODE-based methods provide good estimates in the CONST scenario, with posterior modes centered at the true value and higher posterior densities at truth when compared with the prior. 
In the SD and NM scenarios with the time varying $R_0(t)$, the posterior estimates from the LNA-based method and ODE-based method, though still centered at the truth, do not differ much from the prior distribution.

Posterior summaries of $S(t)$ and $I(t)$ are depicted in the third and fourth rows of Figure~\ref{f:simulation2}.
The two methods produce similar results in the CONST and SD scenario, as both of them have narrow BCIs covering the true trajectories. 
However, in the NM case, while the LNA-based method manages to recover the latent SIR trajectory trend, the BCIs from the ODE-based method fail to cover the true prevalence trajectory in the middle and at the end of the epidemic. Somewhat counterintuitively, LNA-based method produces BCIs for the latent trajectories, $S(t)$ and $I(t)$, that are narrower than its ODE counterparts. 
We suspect this is a result of the ODE-based method poor estimation of the basic reproduction number trajectory at the end of the epidemic.

 \subsection{Frequentist properties of posterior summaries}
 \label{s:repeat}
In this Section, we design a simulation study based on repeatedly simulating SIR trajectories using MJP with pre-specified parameters. 
The simulations are based on the non-monotonic $R_0(t)$ trajectory scenario  in Section \ref{s:oneReal} with the same parameter setup, except the parameters of the lognormal prior for the initial $R_0$ are set to $a_1= 0.7, b_1 = 0.3$. 
Simulating SIR dynamics under low initial number of infected individuals $I_0$ can end up with  low prevalence trajectories that end at the beginning of the epidemic, or trajectories having unrealistically high prevalence, which are less likely to be observed during real infectious disease outbreaks. 
Therefore, while simulating SIR trajectories we reject such ``unreasonable'' realizations to arrive at 100 simulated trajectories. 
The details {of the rejection criteria are given in Section \ref{s:repdetails} of the Appendix. 
For each simulated SIR trajectory, a realization of a genealogy is generated using the structured coalescent process. 
We use both LNA-based and ODE-based model to approximate the posterior distribution of model parameters and latent variables for each genealogy. 
 
We use three metrics to evaluate models based on their estimates of $R_0(t)$ and $I(t)$: average error of point estimates (posterior medians), width of credible intervals, and frequentist coverage of credible intervals. 
Since the value of $R_0(t)$ is greater than 0 and usually upper-bounded by 20 (i.e, it stays within the same order of magnitude), we will measure accuracy using an unnormalized mean absolute error (MAE):
 \begin{eqnarray}
 \mbox{MAE} = \dfrac{1}{T+1}\sum_{i=0}^{T} |\hat{R}_{0_{i}} - R_0(t_i)|,
 \end{eqnarray}
 where $\hat{R}_{0_i}$ is the posterior median of $R_0(t_i)$. 
In contrast, $I(t)$ varies from one at the beginning of the epidemic to thousands at the peak, so to evaluate accuracy of prevalence estimation we use the mean relative absolute error (MRAE):
 \begin{eqnarray}
 \mbox{MRAE} = \dfrac{1}{T+1}\sum_{i=0}^{T}\dfrac{|\hat{I}_i - I(t_i)|}{I(t_i) + 1},
 \end{eqnarray}
 where $ \hat{I}_i$ is the posterior median of I$(t_i)$. 
 We assess precision of $R_0(t)$ estimation based on the mean credible interval width (MCIW): 
 \begin{eqnarray}
 \mbox{MCIW} = \dfrac{1}{T+1}\sum_{i=0}^{T}\left[\hat{R}^{0.975}_{0_{i}} - \hat{R}^{0.025}_{0_{i}}\right], 
 \end{eqnarray}
 where $\hat{R}^{0.025}_{0_{i}}$ and $\hat{R}^{0.975}_{0_{i}}$ denote the lower and upper bounds of the 95\% BCI for $R_{0_i}$. 
 Similar as our measure of accuracy, precision of $I(t)$ estimation is quantified via mean relative credible interval width (MRCIW):
 \begin{eqnarray}
 \mbox{MRCIW} = \dfrac{1}{T+1}\sum_{i = 0}^{T} \dfrac{\hat{I}^{0.975}_i - \hat{I}^{0.025}_i }{I(t_i) + 1},
 \end{eqnarray}
 where $\hat{I}^{0.025}_i$ and $\hat{I}^{0.975}_i$ specify the lower and upper bounds of the $95\%$ BCI of $I(t_i)$. 
In addition, we compute the ``envelope" (ENV) --- a measure of coverage of BCIs the true trajectory --- for $R_0(t)$ and $I(t)$ as follows:
 \begin{eqnarray}
 \mbox{ENV-R}_0 = \dfrac{1}{T+1}\sum_{i=0}^{T}\mathbbm{1}{\left(\hat{R}^{0.025}_{0_{i}} \leq  R_0(t_i) \leq \hat{R}^{0.975}_{0_{i}} \right)},
 \mbox{ENV-I} = \dfrac{1}{T+1}\sum_{i=0}^{T}\mathbbm{1}{\left(\hat{I}^{0.025}_i \leq  I(t_i) \leq \hat{I}^{0.975}_i\right)}. 
 \end{eqnarray}

Sampling distribution boxplots of $R_0(t)$ posterior summaries are depicted in the left three plots of Figure~\ref{f:repeat}.
The LNA-based method yields significantly lower MAE compared with the ODE-based method.  
As a trade-off, the MCIWs produced by the LNA-method are generally higher, as expected since the LNA-based method incorporates the stochasticity in the population dynamics. 
With less bias and wider BCIs, the LNA-based method BCIs result in better coverage than ODE-based BCIs, as shown by the envelope boxplots. 
\par
Sampling distribution boxplots of $I(t)$ posterior summaries, shown in Figure \ref{f:repeat}, are similar to the $R_0(t)$ results, with the LNA-based method generally having lower MRAEs, higher MRCIWs and a better coverage/envelope than the ODE-based method. 
Again, somewhat counterintuitively, the MRCIWs for the LNA-based method are  smaller than the ODE counterparts.
This is likely  caused by  significant bias in $R_0(t)$ estimation by the ODE-based method. 

 We also report the absolute error (AE) and 95\% BCI widths for removal rate $\gamma$ in Figure~\ref{f:repeat}.
We note that an informative prior has been chosen for $\gamma$, because this parameter is weekly identifiably from genetic data alone.
The LNA-based method yields a slightly higher AE than the ODE method.
Both methods produce similar BCI widths.

 \begin{figure}
 	\centerline{\includegraphics[width=\textwidth]{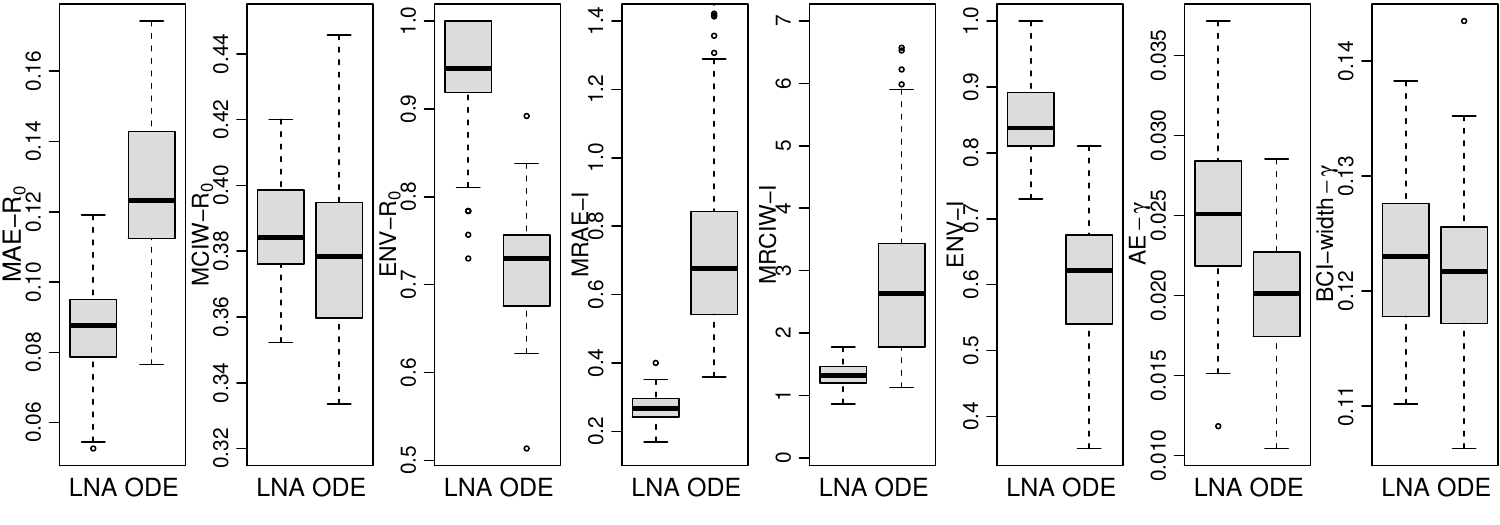}}
 	\caption{Boxplots comparing performance of LNA-based and ODE-based methods using 100 simulated genealogies. The first three plots show mean absolute error (MAE), mean credible interval width (MCIW) and envelope for $R_0(t)$ trajectory. The next three plots depict mean relative absolute error (MRAE), mean relative credible interval width (MRCIW), and envelope for $I(t)$ (prevalence) trajectory. The last two plots show the absolute error (AE) and Bayesian credible interval (BCI) width for $\gamma$. 
}
 	\label{f:repeat}
 \end{figure}
 
 
\section{Analysis of Ebola outbreak in West Africa}
\label{s:realdata}
We apply our LNA-based method to the Ebola genealogies reconstructed from molecular data collected in Sierra Leone and Liberia during the 2014--2015 epidemic in West Africa \citep{dudas2017virus}. 
We use a Sierra Leone genealogy, depicted in the top left plot of Figure~\ref{f:ebolaSL}, which was estimated from 1010 Ebola virus full genomes sampled from 2014-05-25 to 2015-09-12 in 15 cities. 
The Liberia genealogy, shown in the top left plot of Figure~\ref{f:ebolaLib}, was estimated from a smaller number of samples: 205 Ebola virus full genomes sampled from 2014-06-20 to 2015-02-14. 
The original sequence data and the reconstructed genealogies are publicly available at \url{https://github.com/ebov/space-time}. 

When Ebola virus infections were detected in West Africa in mid-Spring of 2014, various intervention measures were proposed and implemented to change behavior of individuals in the populations through which Ebola was spreading. Border closures, encouragement to reduce individual day-to-day mobility, and recommendations on changing burial practices were among the broad spectrum of interventions attempted by
multiple countries. 
It is reasonable to expect that these intervention measures resulted in lowering the contact rates among members of the populations, which in turn reduced the infection rate, or equivalently the basic reproduction number.

When analyzing the Sierra Leone and Liberia genealogies, we rely on conclusions of \citet{dudas2017virus} and assume the population in each country to be well mixed. Furthermore, we assume Ebola spread to follow SIR dynamics. 
For each country, the population size is specified based on its census population size in 2014, with $N =$ 7,000,000 for Sierra Leone and  $N =$ 4,400,000 for Liberia.  
As in our simulation study, we use the lognormal prior for $R_0$ with $a_1 = 0.7$ and $b_1 = 0.5$ and the lognormal prior for the inverse standard deviation $1/\sigma$ with $a_2= 3, b_2=0.2$. 
Recall that this prior setting ensures that \textit{a priori} $R_0(t)$ stays within a reasonable range of $[0,5]$ with probability 0.9. 
For removal rate $\gamma$, we use an informative lognormal prior with mean 3.4 and variance 0.2 based on previous studies \citep{towers2014temporal}. 
The parameter $1/\gamma$, interpreted as the length of the infectious period, is expected to be 8-18 days for each country \textit{a priori}. 
The total time span for each genealogy is divided evenly into 40 intervals, which results in grid interval lengths, $\Delta t_i$s, to be 12.41 days for Sierra Leone and 6.9 days for Liberia. 
We run the MCMC algorithm in Section \ref{s:inference} for 3,000,000 iterations for Sierra Leone data and 750,000 iterations for Liberia data. 
The posterior samples are obtained by discarding the first 100,000 iterations and saving every 30th iteration afterward. 
The trace plots in Section \ref{s:traceplotEbola} of the Appendix indicate the MCMC algorithm has converged and achieved good mixing  in each case. 

\begin{figure}
	\centerline{\includegraphics[width=0.9\textwidth]{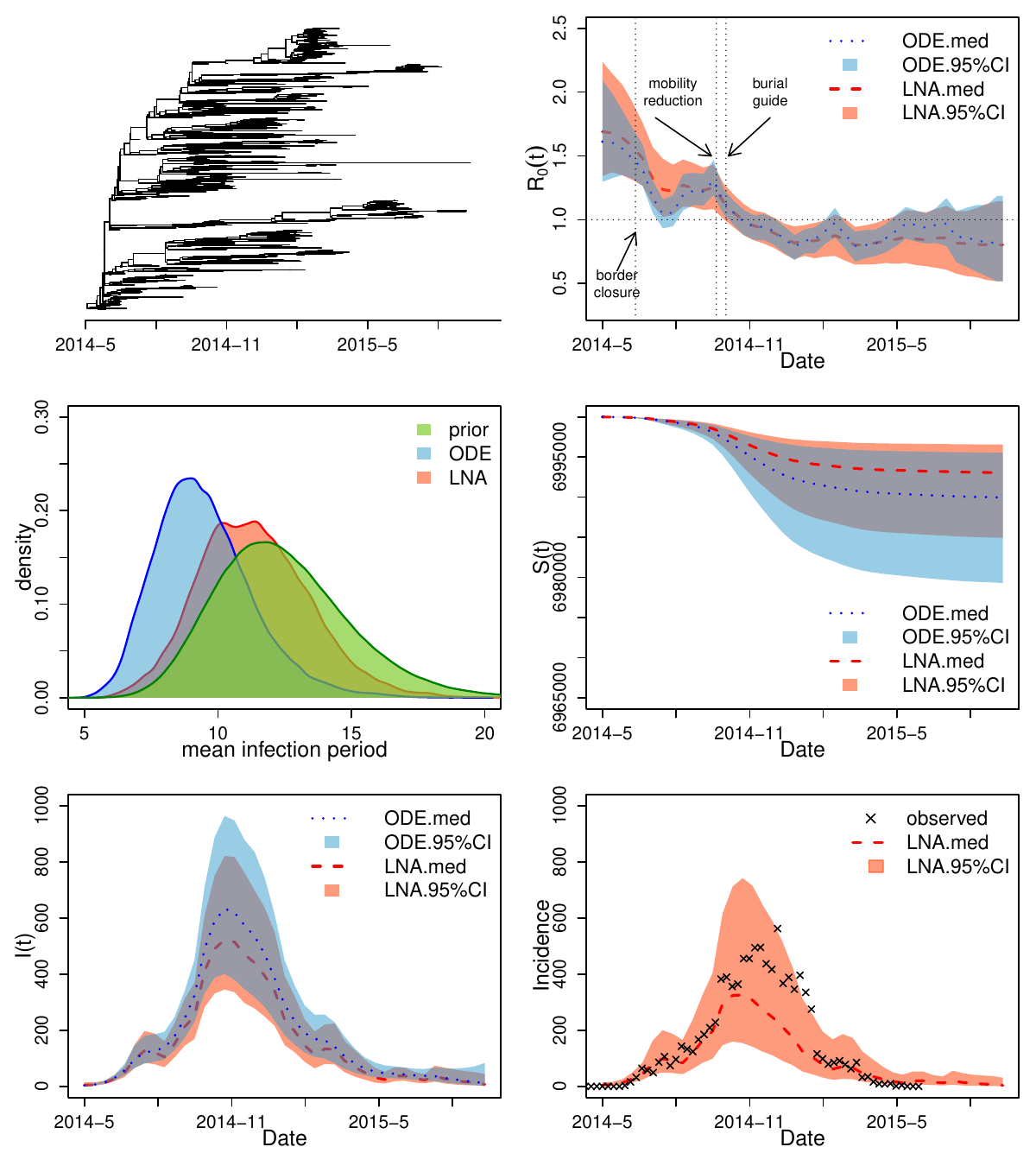}}
	\caption{Analysis of the genealogy relating Ebola virus sequences collected in Sierra Leone. 
		Top top left plot depicts the Ebola genealogy. The top right plot shows the estimated $R_0(t)$, with the red dashed line showing the posterior median and the salmon shaded area showing the 95\% BCIs of  the LNA-based method. The posterior median based on the ODE-based method is plotted as the blue dotted line with blue shading corresponding to the 95\% BCIs. The medium left figure shows prior and posterior densities of the mean infection period $1/\gamma$.  The prior density is shown in green, while the posterior densities based on LNA and ODE are plotted in red and blue respectively. 
The medium right and the bottom left figures show the estimated trajectory of $S(t)$ and $I(t)$, using the same legend as in top right plot. The bottom right plot shows the predicted median and 95\% BCIs for weekly reported incidence together with the reported incidence from WHO shown as crosses.}
	\label{f:ebolaSL}
\end{figure}

Figures \ref{f:ebolaSL} and \ref{f:ebolaLib} show results for Sierra Leone and Liberia respectively, with intervention events mapped onto the calendar time on the x-axis. 
Our LNA-based method estimates the initial $R_0$ in Sierra Leone during 2014--2015 to be 1.68, with 95\% BCI of $(1.33,2.23)$. 
Similarly, $R_0$ in Liberia during 2014--2015 has a point estimate 1.67 and a 95\% BCI $(1.29,2.24)$. 
Our estimate of initial $R_0$ in Sierra Leone is consistent with the estimates of \citet{stadler2014insights}, 
who fitted multiple birth-death models to 72 sequences at the early stages of the outbreak.  
\citet{volz2014phylodynamic} used a susceptible-exposed-infectious-recovered (SEIR) model with a constant $R_0$ and estimated it to be 2.40 (CI: $(1.54,3.87)$).
\citet{althaus2014estimating} assumed an exponentially decaying $R_0(t)$ with an estimated initial $R_0$ of 2.52 (CI: $(2.41, 2.67)$). 
The discrepancies between our and SEIR-based estimates are not unexpected, because SEIR models generally yield higher $R_0$ estimates than SIR models when applied to the same dataset \citep{wearing2005appropriate,keeling2011modeling}. 
Our estimated $R_0$ for Liberia is in agreement with results of \citet{althaus2014estimating}, who fitted a SEIR model to
incidence data and arrived at an estimated $R_0$ of 1.59 (CI: $(1.57,1.60)$). 
\par
The $R_0(t)$ dynamics in the two countries share a similar pattern: with (1) a decreasing trend that starts in Spring/Summer of 2014, (2) a stable/constant period until the end of September 2014 and (3) a final decrease below 1.0 (epidemic is contained) around November 2014. 
Since the number of susceptible individuals did not change significantly over the course of the epidemic, relative to the total population size, the basic and effective reproduction numbers, $R_0(t) = \beta(t) N/\gamma$ and $R_{\text{eff}}(t) = \beta(t) S(t)/\gamma$, are approximately equal. This allows us to compare our $R_0(t)$ estimation results with previously estimated changes in $R_{\text{eff}}(t)$.
Our estimation of early $R_0(t)$ dynamics in Sierra Leone agrees with results of \citet{stadler2013birth}, who concluded that the effective reproduction number did not significantly decrease until mid June. 
Our estimated $R_0(t)$ trajectory suggests that later interventions, such as border closures and release of burial guides, may have been helpful in controlling the spread of the disease. 
The infectious period for Sierra Leone epidemic is estimated to be 11.2 days with a 95\% BCI (7.6,16). 
For Liberia, the infection period has a point estimate of 9.8, with a 95\% BCI $(6.87,14.05)$. 
The posterior median of the total number of infected individuals (final epidemic size) is 7,284 and its 95\% BCI is $(3397,14870)$ for Sierra Leone, which is close to 8,706 total confirmed number of cases reported by \citet{cdc2016eboladat} (CDC). 
Liberia had a smaller epidemic than Sierra Leone, with estimated total infected individuals being 2,842 and a 95\% BCI of $(1296,6173)$.  These results are also in agreement with 3,163 total confirmed cases from CDC. 
\par
We perform an out-of-sample validation by comparing our results with weekly reported confirmed incidence in Sierra Leone and Liberia from the \citet{who2016eboladat} (WHO). 
The posterior predictive weekly incidence at time $t$, denoted by $\hat{N}(t)$, is approximated by
\begin{eqnarray}
\hat{N}(t) = \hat{\beta}(t)\hat{S}(t)\hat{I}(t) \cdot \Delta t, 
\end{eqnarray}
where $\hat{\beta}(t),\hat{S}(t)$ and $\hat{I}(t)$ are the posterior estimates of the infection rate, number of susceptible and number of infected individuals at time $t$ respectively, and $\Delta t := 7/365$ corresponds the time interval of one week. 
We plot the posterior predictive estimates of weekly incidence together with the corresponding weekly reported confirmed incidence. 
For both countries, our model-based incidence 95\% BCIs cover the reported incidence counts from WHO, suggesting that our time varying SIR model can estimate incidence well from genetic data alone. 
We note that our estimated latent incidence should be greater than the reported incidence, because not all Ebola cases were reported and recorded. 
However, the discrepancy between latent and reported incidence should not be large, because Ebola reporting rate was high. 
For example, \citet{scarpino2014epidemiological} estimated that 83\% of Ebola cases were reported. 
\par
We also report results from the ODE-based method and superimpose these results over LNA-based results on Figures \ref{f:ebolaSL} and \ref{f:ebolaLib}. 
For the relatively small Liberia genealogy, the ODE-based and LNA-based methods yield similar parameter estimates. 
However,  the larger Sierra Leone genealogy produces substantial differences between ODE-based and LNA-based estimates of the $R_0(t)$. 
The ODE-based method captures the decreasing trend of $R_0(t)$ in Spring and Summer of 2014, but provides narrow BCIs with unrealistic short term fluctuations in the basic reproduction number trajectory. 

\begin{figure}
	\centerline{\includegraphics[width=0.9\textwidth]{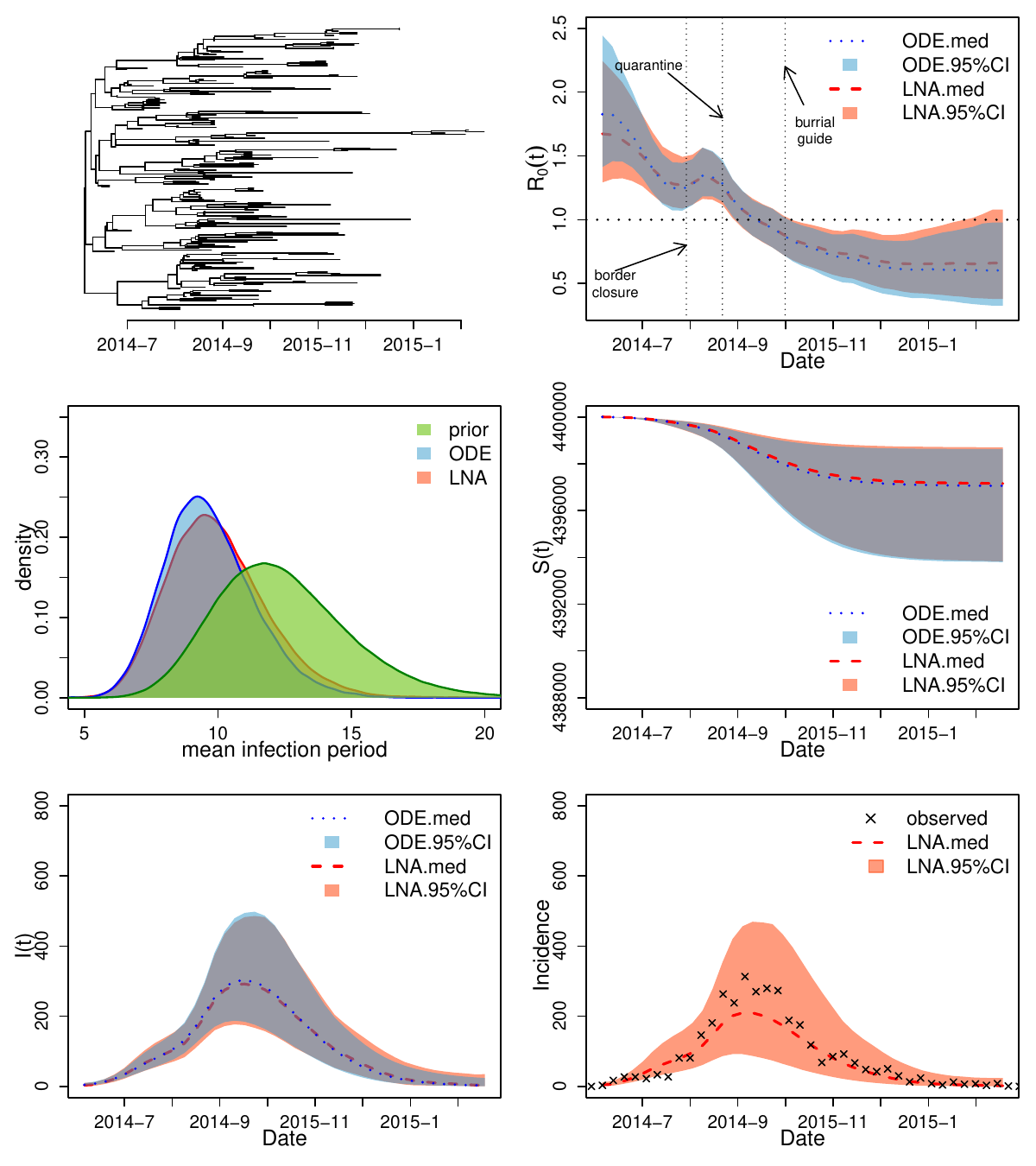}}
	\caption{Analysis of the genealogy relating Ebola virus sequences collected in Liberia. See caption in Figure~\ref{f:ebolaSL} for the explanation of the plots.}
	\label{f:ebolaLib}
\end{figure}

\section{Discussion}
\label{s:discuss}
In this paper, we propose a Bayesian phylodynamic inference method that can fit a stochastic epidemic model to an observed genealogy estimated from infectious disease genetic sequences sampled during an outbreak. 
Our statistical model can be viewed as semi-parametric: with (1) a parametric SIR model describing the infectious disease dynamics and (2) a non-parametric GMRF-based estimation of the time varying basic reproduction number. 
To the best of our knowledge, this is the first method combining a Bayesian nonparametric approach with a deterministic or stochastic SIR model for phylodynamic inference (although see \citep{xu2016bayesian} for a similar approach applied to more traditional epidemiological data). 
Our use of LNA allows us to devise an efficient MCMC algorithm to approximate high dimensional posterior distribution of model parameters and latent variables. 
Our LNA-based method produces posterior summaries with better frequentist properties than the state-of-the-art ODE-based method, underscoring the importance of working with stochastic models even in large populations.
We showcase our method by applying it to the Ebola genealogies estimated from viral sequences collected in Sierra Leone and Liberia during the 2014--2015 outbreak. Our nonparametric estimates of $R_0(t)$ in Sierra Lione and Liberia suggest that the basic reproduction number decreased in two-stages, where the second stage brought it below 1.0 --- a sign of epidemic containment. The second stage of $R_0(t)$ decrease closely follows in time implementation of interventions, pointing to their effectiveness.

Our method relies on the assumption that population dynamics follow a SIR model.
However, it may be desirable to be able to relax this assumption. 
For example, in Ebola spread modeling some authors used a SEIR model that assumes a latent period during which an infected individual is not infectious  \citep{althaus2014estimating,volz2018bayesian}. One future direction of this work is to generalize the LNA-based method to fit complicated compartmental epidemic models, 
including models with multi-stage infections like SEIR model and models with the population stratified by sex, age, geographic location, or other demographic variables. 
The structured coalescent likelihoods under these models may not have closed-form expressions. 
However, \citet{volz2012complex} and \citet{muller2017structured} propose several strategies to  approximate structured coalescent likelihoods. 
It would be interesting to combine these approximation strategies with LNA to estimate  parameters of complex stochastic epidemic models from genealogies.

The experiments in Section \ref{s:oneReal} and Appendix Section \ref{s:sensitivity} indicate that one has to pay close attention to parameter identifiability when fitting SIR models to genealogies or to sequence data directly. 
Identifiability may not be a problem under an assumption of a constant $R_0(t)$. 
However, the removal rate tends to be only weakly identifiable in the scenarios with a time-varying basic reproduction number, in which the estimation can be sensitive to the choice of priors. 
In Section \ref{s:sensitivity} of the Appendix, we demonstrate that putting a weakly informative prior on the removal rate can cause bias not only in the estimation for removal rate, but also can lead to a failure in recovering the reproduction number and latent population dynamics. 
Therefore,  successful inference of SIR model parameters using genealogical data should rely on a sound informative prior for the removal rate. 
This constraint is not a big shortcoming in practice, since prior information about the removal rate, or mean length of the infectious period, is usually readily available from patient hospitalization data \citep{who2014ebola}.

Since parameter identifiability is a recurring problem in infectious disease modeling, integration of multiple sources of information is of great interest. 
Using particle filter MCMC, \citet{rasmussen2011inference}  demonstrated that jointly analyzing genealogy and incidence case counts considerably reduces the uncertainty in both estimation of latent population trajectory and SIR model parameters, compared with estimation based on a single source of information.
We plan to use our LNA-based framework to perform similar integration of genealogical data and incidence time series. 
Another possible source of information is the distribution of genetic sequence sampling times. 
\citet{karcher2016quantifying} proposed a preferential sampling approach that explicitly models dependence of the sampling times distribution on the effective population size. 
The authors demonstrated that accounting for preferential sampling helps decrease bias and results in more precise effective population size estimation.
It would be interesting to incorporate preferential sampling into our LNA-based framework by assuming a probabilistic dependency between sampling times and latent prevalence $I(t)$. 
\par
Our method assumes a genealogy/phylogenetic tree is given to us. 
In reality, genealogies are not directly observed and need to be inferred from molecular sequences. Ideally, uncertainty in the genealogy should be handled by building a Bayesian hierarchical model and integrating over the space of genealogies using MCMC. 
In fact, implementations of such Bayesian hierarchical modeling already exist for nonparametric, birth-death, and ODE-based phylodynamic approaches \citep{drummond2005bayesian,minin2008smooth, gill2013, stadler2013birth,volz2018bayesian}.
Therefore, an important future direction will be to extend our LNA framework to fitting stochastic epidemic models to molecular sequences instead of genealogies.
Similarly to the structured coalescent model implementation of \citet{volz2018bayesian}, the easiest way to achieve this will be integration of our LNA MCMC algorithm into popular open source phylogenetic/phylodynamic software packages, such as \texttt{BEAST}, \texttt{BEAST2}, and \texttt{RevBayes}} \citep{suchard2018bayesian, Bouckaert2014, Hohna2016}.

\section*{Acknowledgements}
We thank Jon Fintz for discussing the details of implementing Linear Noise Approximation. We are grateful to Michael Karcher and Julia Palacios for patiently answering questions about nonparametric phylodynamic inference and \texttt{phylodyn} package.  
M.T and V.N.M.\ were supported by the NIH grant R01 AI107034.
M.T., T.B., and V.N.M.\ were supported by the NIH grant U54 GM111274.
G.D. was supported by the Mahan postdoctoral fellowship from the Fred Hutchinson Cancer Research Center. 
This work was supported by NIH grant R35 GM119774-01 from the NIGMS to TB. TB is a Pew Biomedical Scholar.






\bibliography{biomsample_bib}
\newpage
\setcounter{table}{0}
\renewcommand{\thetable}{A-\arabic{table}}
\renewcommand{\thefigure}{A-\arabic{figure}}
\renewcommand{\thesection}{A-\arabic{section}}

\renewcommand{\theequation}{A-\arabic{equation}}
\setcounter{equation}{0}
\setcounter{section}{0}
\setcounter{figure}{0}
\begin{center}
	\LARGE{Appendix} \end{center}

\section{A general framework for stochastic kenetic models}
\label{s:generalization}
\subsection{Stochastic model generalization}
In Section \ref{s:method}, we provide an example of the linear noise approximation (LNA) for the SIR model. The LNA framework can be also generalized to other types of the stochastic kinetic models in Infectious Disease Epidemiology and in Systems Biology. 
Here, we give a general representation of the stochastic kinetic model by viewing it as a reaction network system. The notation is based on the work of \citet{fearnhead2014inference}. 

Let's start with a reaction system with $d$ reactants $\mathcal{X}_1,\ldots,\mathcal{X}_d$ and $q$ reactions. Without loss of generality, each reaction is assumed to have a constant rate parameter $\theta_i$ for $i = 1,\ldots,q$ and $\btheta = (\theta_1,\dots,\theta_q)$ denotes the rate vector of the system (this framework can be extended to handle stochastic kinetic models with time-varying rates as in Section \ref{s:method} of the main text). 
The transition event in the $i$th reaction ($i=1,\ldots,q$) has the following form:
\begin{eqnarray}
\label{eq:reaction}
\tilde{a}_{i1}\mathcal{X}_1 + \cdots + \tilde{a}_{id}\mathcal{X}_d\xrightarrow{\theta_i}  \bar{a}_{i1} \mathcal{X}_1 + \ldots +\bar{a}_{i1} \mathcal{X}_d,
\end{eqnarray}
where $\tilde{a}_{ij}$ and $\bar{a}_{ij}$ are non-negative integers representing the number of $\mathcal{X}_j$ in the $i$th reaction equation. 
In a compartmental stochastic epidemic model, the coefficient $\tilde{a}_{ij}$ will be either 0 or 1. The transitions in the reaction system can be encoded in an effect matrix,
\begin{eqnarray}
\bA := \{\tilde{a}_{ij} - \bar{a}_{ij}\} \in \mathbb{Z}^{q\times d}, 
\end{eqnarray}
with each row corresponding to a certain type of reaction event and each column representing the change in the counts of reactants. 
Let $X_j(t)$ denote counts/population of the $\mathcal{X}_j$ at $t$, and the population state at time t can be tracked by vector $\bX(t):= (X_1(t),\ldots,X_d(t))$. Let $h_i$ denote the reaction rate of the $i$th reaction, where $h_i$ can be written as 
\begin{eqnarray}
h_i =  \theta_i\prod_{j=1}^{d}\binom{X_j}{\tilde{a}_{ij}}.
\end{eqnarray}
Hence, following the same notation as in Section~\ref{s:SIR} of the main text, the rate vector $\bh$ and the rate matrix $\bH$ can be defined as 
\begin{eqnarray}
\bh(\bX,\btheta) = \left(h_1,\ldots, h_q\right)^T, \qquad \bH(\bX,\btheta)= \mbox{diag}\left(\bh(\bX,\btheta)\right). 
\end{eqnarray}

Given the above notation, the deterministic ordinary differential equation model of the reaction system can be written as 
\begin{eqnarray}
\df \bX = \bA^T \bh\left(\bX,\btheta\right)\df t, \qquad \bX(0) = \bx_0,
\end{eqnarray}
where $\bx_0$ is a vector of initial counts of reactants $\mathcal{X}_1,\ldots, \mathcal{X}_d$. 

\subsubsection{Example: SEIR model}
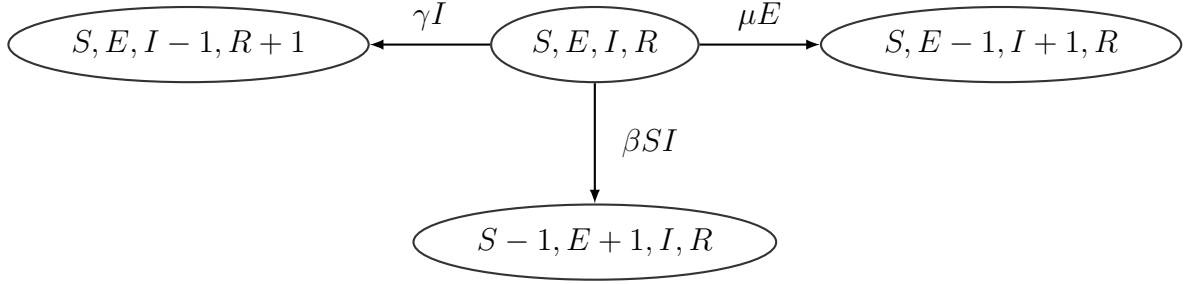
\begin{figure}
	\centering
	\begin{tikzpicture}
	\usetikzlibrary{shapes}
	\tikzstyle{main}=[ellipse, minimum size = 10mm, thick, draw =black!80, node distance = 16mm]
	\tikzstyle{connect}=[-latex, thick]
	\tikzstyle{box}=[rectangle, draw=black!100]
	\node[main] (theta)  {$S,E,I-1,R+1$};
	\node[main] (z) [right=of theta] {$S,E,I,R$};
	\node[main] (w) [right=of z] { $S,E -1,I+1, R$};
	\node[main] (y) [below=of z] { $S-1, E+1, I, R$};
	\path (z) edge [ connect] node [midway,above=0.1em]  {$\gamma I$}(theta)
	(z) edge [connect] node [midway,above=0.1em]  {$\mu E$}(w)
	(z) edge [connect] node [midway,right=0.5em]  {$\beta SI$}(y);
	\end{tikzpicture}
	\caption{SEIR Markov jump process. From the current state with the counts $S,E,I,R$, the population can transition to (1) state $S-1,E+1, I, R$ (an infection event) with rate $\beta SI$ or  to (2) state $S, E - 1, I + 1, R$ (an event where infected individual becomes infectious) with rate $\mu E$ or to (3) state $S,E,I-1,R+1$ (a removal event) with rate $\gamma I$. No other instantaneous transitions are allowed.}
	\label{f:seir}
\end{figure}

The above general representation of stochastic kinetic models can be directly applied to stochastic epidemic models. 
Here, we illustrate this on a Susceptible-Exposed-Infected-Recovery (SEIR) model. 
SEIR model is an extension of the SIR model that assumes a latent period called ``Exposed'', in which an infected individual does not have the ability to infect others. 
The exposed individual will eventually become infectious with rate $\mu$. As in the SIR model, an infectious individual has removal/recovery rate $\gamma$. 
The transition events between different states of the SEIR model are depicted in Figure~\ref{f:seir}. 

Following the stochastic kinetic model representation, the SEIR model can be viewed as a reaction system of four reactants --- susceptible, exposed, infected, and recovered individuals --- and the following three ``reactions": 
\begin{eqnarray}
&\mbox{Susceptible} + \mbox{Infected}\xrightarrow{\beta} \mbox{Exposed} + \mbox{Infected} \label{eq:reaction3.1},\\
& \mbox{Exposed} \xrightarrow{\mu} \mbox{Infected} \label{eq:reaction3.2},\\
& \mbox{Infected} \xrightarrow{\gamma} \mbox{Recovered}\label{eq:reaction3.3}.
\end{eqnarray}
Since the recovered population never interacts with individuals in other compartments, we will only keep track of the counts of susceptible, exposed, and infectious individuals at time $t$, denoted by $S(t)$, $E(t)$, $I(t)$ respectively. 
The effect matrix $\bA$ for the SEIR model can be written as:
\begin{equation}
\label{eq:effectmx3}
\bA = \begin{blockarray}{cccc}
\mbox{Susceptible} & \mbox{Exposed} &\mbox{Infected}\\
\begin{block}{(ccc)c}
-1 & 1 & 0 & \mbox{reaction }\eqref{eq:reaction3.1}\\
0 & -1& 1 & \mbox{reaction }\eqref{eq:reaction3.2}\\
0 & 0& -1 &\mbox{reaction }\eqref{eq:reaction3.3}\\
\end{block}
\end{blockarray},
\end{equation}
with columns representing compartments and rows representing reactant changes during reaction events. 

If we let $\bX(t) = (S(t),E(t),I(t))$ denote the state vector at time $t$, 
then the rate vector $\bh$ for the SEIR model is 
\begin{eqnarray}
\bh(\bX(t),\btheta) = (\beta S(t)I(t), \mu E(t), \gamma I(t))^T. 
\end{eqnarray}

\section{Derivation of the linear noise approximation}
\label{s:LNAderive}
\subsection{SDE approximation for MJP}
A stochastic way to approximate the MJP model is to use the Stochastic Differential Equation (SDE) approximation, also known as the chemical Langevin equation (CLE) \citep{gillespie2000chemical}. 
The SDE method can be viewed an approximation of the MJP at time $t$, obtained by applying a normal approximation to the Poisson distributed number of state transitions in a small interval of time $(t,t+\Delta t)$ \citep{gillespie2000chemical,wallace2010simplified}. 
The deterministic part in SDE corresponds to the right hand side of ODE \eqref{eq:ODEeta} and stochastic part is related to the variance of the system.
The SDE for general stochastic kinetic models can be written as 
\begin{eqnarray}
\label{eq:SDE}
\df \bX(t) = \bA^T \bh(\bX(t),\btheta(t)) \df t + \sqrt{\bA^T \bH(\bX(t),\btheta(t))\bA } \cdot \df \bW_t,
\end{eqnarray}
where $\bW_t$ denote a $d$ dimensional Wiener process and the square root $\sqrt{\cdot}$ means the Cholesky triangle of the $d\times d$ covariance matrix. 
\subsection{LNA approximation of the SDE}
Since in the main text we assume the rate $\btheta(t)$ varies in a piecewise constant way,  without loss of generality, we use the notation $\btheta$ for the rate in a given time interval where the LNA is applied. 
\begin{theorem}[Linear Noise Approximation for SDE]
	\label{thm:LNA}
	Let $\etab(t)$ be the solution of ordinary differential equation \eqref{eq:ODEeta} with initial value $\etab_0$. Let $N$ be the system size, which is the total number of individuals in the system (In SIR model, $N$ will be the total population, i.e $N = S + I + R$), $\btheta = (\theta_1,\ldots,\theta_q)$ denote the vector of rate parameters in $q$ reactions. Then the solution $\bX(t)$ of the SDE \eqref{eq:SDE} satisfies the following equation
	\begin{equation}
	\begin{split}
	\dfrac{1}{\sqrt{N}}\df (\bX(t)-\etab(t)) =\dfrac{1}{\sqrt{N}}\left(\bF(\etab(t),\btheta)(\bX(t)-\etab(t)) + o(1) \right)\df t + \\
	\left( \dfrac{1}{\sqrt{N}} \sqrt{\bA^T \bH(\etab(t),\btheta)\bA} +o(1)\right) \df \bW_t,
	\end{split}
	\label{eq:Taylor}
	\end{equation}
	as $N \rightarrow +\infty$.
\end{theorem}

\begin{proof}
	The following derivation is based on \citep{wallace2010simplified}. 
	
	We rescale both the compartment size and reaction rates as follows:
	\begin{eqnarray}
	\tilde{\bX}(t) &=& N^{-1}\cdot \bX(t)\\
	\tilde{\theta}_i & =&  N^{m_i-1}\theta_i,
	\end{eqnarray}
	where $m_i = \sum_{j = 1}^{d} \tilde{a}_{ij}$ is the sum of coeffcients in the left hand side of $i$-th reaction as in Section~\ref{s:generalization}. 
	The transformed $\tilde{\bX}(t)$ represents the proportion of individuals/reactants each compartment with respect to the total population size. 
	Then we have $\bh(\bX(t),\btheta) = N\bh(\btX(t),\bttheta) $ and $\bF(\etab(t),\btheta) = \bF(\teta(t),\bttheta) $. Hence, the SDE \eqref{eq:SDE} becomes 
	\begin{eqnarray}
	\label{eq:SDE2}
	\df \btX(t) = \bA^T\bh(\btX(t),\bttheta) \df t + \dfrac{1}{\sqrt{N}}\sqrt{\bA^T \bH(\btX(t),\bttheta)\bA}\cdot \df \bW_t .
	\end{eqnarray}
	Let $\teta(t)$ be the solution of the ODE
	\begin{eqnarray}
	\label{eq:eta_tilde}
	\df \teta(t) = \bA^T\bh (\teta(t),\bttheta) \df t,
	\end{eqnarray} 
	and  we have 
	$\etab(t) = N \teta(t)$, where $\etab(t)$ is the solution of the ODE \eqref{eq:ODEeta}. $\teta(t)$ can be viewed as a scaled version solution of ODE \eqref{eq:ODEeta}. 
	Let $\bxi(t) = \sqrt{N}\left(\btX(t) - \teta(t)\right) = \dfrac{1}{\sqrt{N}}\left(\bX(t) - \etab(t)\right)$ denote the scaled residual, then the rescaled compartment size vector $\btX(t)$ can be written as 
	\begin{eqnarray}
	\btX(t) = \dfrac{1}{\sqrt{N}}\bxi(t) + \teta(t). 
	\end{eqnarray}
	After using first order Taylor expansion of $\bh(\btX(t),\bttheta)$ and $\bH(\btX(t),\bttheta)$ around $\btX=\teta(t)$, the SDE \eqref{eq:SDE2} becomes 
	\begin{eqnarray*}
		\df \btX(t) &=& \bA^T\bh\left(\teta(t) + \dfrac{1}{\sqrt{N}}\bxi(t) ,\bttheta\right)\df t + \sqrt{\bA^T\bH\left(\teta(t) + \dfrac{1}{\sqrt{N}}\bxi(t),\bttheta\right)\bA} \cdot \df\bW_t \\
		& = &\left(\bA^T\bh(\teta(t),\bttheta) +  \bF(\teta(t),\bttheta)\cdot \dfrac{1}{\sqrt{N}}\bxi(t) + \mathcal{O}(N^{-1})\right) \df t \\
		&&+ \dfrac{1}{\sqrt{N}}\sqrt{ \bA^T \bH(\teta(t),\bttheta)\bA + \mathcal{O}(\dfrac{1}{\sqrt{N}})}\cdot \df\bW_t\\
		& =  &\left(\bA^T\bh(\teta(t),\bttheta) +  \dfrac{1}{\sqrt{N}}\bF(\teta(t),\bttheta)\cdot \bxi(t) \right)\df t \\
		& &+  \dfrac{1}{\sqrt{N}}\sqrt{ \bA^T \bH(\teta(t),\bttheta)\bA} \cdot \df\bW_t + o(N^{-1/2})\df\bW_t + o(N^{-1})\df t.
	\end{eqnarray*}
	where $\bF(\teta(t),\btheta):= \dfrac{\partial \bA^T\bh(\btX(t),\btheta)}{\partial \btX} \Big |_{\btX = \teta(t)}$ is the Jacobian matrix of the deterministic part $\bA^T\bh\left(\btX(t),\btheta\right)$ in \eqref{eq:ODEeta} at $\teta(t)$.
	By subtracting \eqref{eq:eta_tilde} and multiplying by $\sqrt{N}$ on the two ends, the above equation becomes a differential equation with respect to $\bxi$:
	\begin{eqnarray}
	\df \bxi(t)  = \bF(\teta(t),\btheta)\bxi(t) \df t + \sqrt{ \bA^T \bH(\teta(t),\bttheta)\bA}\cdot \df\bW_t + o(N^{-1/2})\df\bW_t + o(N^{-1})\df t.
	\end{eqnarray}
	After multiplying by $\sqrt{N}$, the above equation gives us \eqref{eq:Taylor}. 
\end{proof}
Recall that $\bM(t)$ is the solution of \eqref{eq:bM} with initial condition $\bM(0) = \bX_0 - \etab_0$. We can use $\etab(t) + \bM(t)$ as an approximation of $\bX(t)$. 
Based on the local Lipschitz property of $\bF(\etab(t),\btheta)$ with respect to $t$ and $\bA^T\bH(\etab(t),\btheta)$, 
$\bX(t)$ can be approximated by $\etab(t) + \bM(t)$  with 
\begin{eqnarray}
\bX(t)= \etab(t) + \bM(t) + o(N^{1/2}),
\end{eqnarray}
for fixed $t$ as system size $N\rightarrow +\infty$.

%

\subsection{Derivation of equations \eqref{eq:m} and \eqref{eq:Phi} in the main text}
\label{s:deriveM}
\begin{lemma} [Solution of linear ODE system] \label{thm:linearODE}
Let $\bF(t)\in \mathbb{R}^{d\times d}$ and $\bX(t)\in \mathbb{R}^{d}$ be function of defined on $\{t: t\geq 0\}$ that satisfies the following linear ODE
	\begin{eqnarray}
	\label{eq:ODElemma}
	\df \bX(t) = \bF(t) \bX(t) \df t, \qquad \bX_0 = \bx_0. 
	\end{eqnarray} 
	For $t\geq 0$, the solution of (\ref{eq:ODElemma}) can be represented as
	\begin{eqnarray}
	\bX(t) = \bSigma(t,0) \bx_0
	\end{eqnarray}
	where $\bSigma(t,0)$ is the solution of ordinary differential equation in $\mathbb{R}^{d\times d}$
	\begin{eqnarray}
	\df \bSigma(t,0) = \bF(t)\bSigma(t,0)\df t, \qquad \bSigma(0,0) = \bI. 
	\end{eqnarray}
\end{lemma}

Lemma (\ref{thm:linearODE}) gives the solution of linear ODE. Hence, the solution for the main text linear ODE \ref{eq:m} is on $[t_{i-1},t]$ will be 
\begin{eqnarray}
\bbm(t) = \bSigma(t,t_{i-1})\bbm_{i-1}, 
\end{eqnarray}
where $\bbm_{i-1}$ is the initial state at $t_{i-1}$ and $\bSigma(t,t_{i-1})$ is the transition matrix by 
\begin{eqnarray}
\label{eq:fundamental}
\df \bSigma(t,t_{i-1}) = \bF(\etab(t),\btheta) \bSigma(t,t_{i-1}) \df t, \qquad \bSigma(t_{i-1};t_{i-1}) = \bI, 
\end{eqnarray} 
and $\bbm_{i-1}$ is the initial value for $\bbm$ at time $t_{i-1}$. 
\begin{theorem}
	\label{thm:bM}
	Let $\{\bM(t)\}_{t \geq 0} \in \mathbb{R}^d$ be stochastic process that satisfies the following stochastic differential equation, 
	\begin{eqnarray}
	\label{eq:bM2}
	\df \bM(t) = \bF(\etab(t),\btheta) \bM(t) \df t + \sqrt{\bA^T\bH(\etab(t),\btheta) \bA }\df \bW_t.
	\end{eqnarray} 
	Then the solution of (\ref{eq:bM2}) is the Gaussian process 
	\begin{eqnarray}
	\bM(t) = \bSigma(t,t_0)\left(\bM(t_0) + \int_{t_0}^{t}\bSigma^{-1}(s,t_0)\sqrt{\bA^T\bH(\etab(t),\btheta)\bA}\df \bW_s \right),
	\end{eqnarray}
	with mean process $\bbm(t) := \bE\left[\bM(t)| \bM(t_0)\right]$ satisfies \eqref{eq:m} and variance process $\bPhi(t):=\mbox{Var}\left[\bM(t)| \bM(t_0)\right]$ satisfies \eqref{eq:Phi}. 
\end{theorem}
\begin{proof}
	Define matrix function  $\bSigma(t,t_{0})$ as \eqref{eq:fundamental}. First we apply  the linear transform $\tilde{\bM}(t) = \bSigma^{-1}(t;t_0)\bM(t)$. Based on Ito's lemma, \eqref{eq:bM2} can be simplified as a SDE of $\tilde{\bM}(t)$: 
	\begin{eqnarray}
	\df \tilde{\bM}(t) = \bSigma^{-1}(t;t_0)\sqrt{\bA^T\bH(\etab(t),\btheta)\bA}\df \bW_t, 
	\end{eqnarray}
	with solution. 
	$$\tilde{\bM}(t) = \tilde{\bM}(t_0) +  \int_{t_0}^{t} \bSigma^{-1}(s;t_0)\sqrt{\bA^T\bH(\etab(t),\btheta)\bA} \cdot \df \bW_s$$
	Then the solution of $\bM(t)$ is 
	\begin{eqnarray}
	\label{eq:bMsolve}
	\bM(t) = \bSigma(t,t_0)\left(\bM(t_0) + \int_{t_0}^{t}\bSigma^{-1}(s,t_0)\sqrt{\bA^T\bH(\etab(t),\btheta)\bA}\df \bW_s \right).
	\end{eqnarray}
	$\bSigma(t,t_0)\bM(t_0)$ in \eqref{eq:bMsolve} is a deterministic function of $t$. The integral $\int_{t_0}^{t}\bSigma^{-1}(s,t_0)\sqrt{\bA^T\bH\left(\eta(t)\right)\bA}\df \bW_s$ in \eqref{eq:bMsolve} should be Gaussian random variable with mean $\mathbf{0}$ since it is a
	linear combination of the increments of Brownian motion with different variance. Hence, the $\bM(t)$ should be a Gaussian process.  By taking the expectation of \eqref{eq:bMsolve}, the mean of $\bbm(t) = \bE [\bM_t]$ satisfies
	\begin{eqnarray}
	\bbm(t) = \bSigma(t,t_0) \bbm(t_0),
	\end{eqnarray}
	which corresponds to the solution of ODE \eqref{eq:m}. 
	
	For the variance process, from \eqref{eq:bMsolve}, 
	\begin{eqnarray}
	\label{eq:intPhi}
	\bPhi(t) =\bSigma(t,t_0) \int_{t_0}^{t}  \bSigma^{-1}(s,t_0) \bA^T\bH(\etab(t),\btheta)\bA  \bSigma^{-1}(s,t_0)  \df s \cdot \bSigma^T(t,t_0)
	\end{eqnarray}
	By differentiation with respect to $t$, \eqref{eq:intPhi} becomes 
	\begin{eqnarray*}
		\df \bPhi(t) &= & \df\bSigma(t,t_0)\cdot \int_{t_0}^{t} \bSigma^{-1}(s,t_0) \bA^T\bH(\etab(t),\btheta)\bA  \bSigma^{-T}(s,t_0)\df s\cdot \bSigma^T(t,t_0)  \\
		&+ &\bSigma(t,t_0)\cdot \df\left[\int_{t_0}^{t}\bSigma^{-1}(s,t_0) \bA^T\bH(\etab(t),\btheta)\bA  \bSigma^{-T}(s,t_0) \df s\right]\cdot \bSigma^T(t,t_0) \\
		&+& \bSigma(t,t_0) \cdot \int_{t_0}^{t}\bSigma^{-1}(s,t_0) \bA^T\bH(\etab(t),\btheta)\bA  \bSigma^{-T}(s,t_0) \df s \cdot  \df \bSigma^T(t,t_0) \\
		&=& \bF(\etab(t),\btheta) \cdot\bSigma(t,t_0)\int_{t_0}^{t} \bSigma^{-1}(s,t_0) \bA^T\bH(\etab(t),\btheta)\bA  \bSigma^{-T}(s,t_0)  \df s \cdot \bSigma^T(t,t_0) \df t\\
		&+& \bSigma(t,t_0) \cdot \bSigma^{-1}(t,t_0) \bA^T\bH(\etab(t),\btheta)\bA  \bSigma^{-T}(t,t_0)\cdot \bSigma^T(t,t_0)  \cdot \df t\\
		&+&  \bSigma(t,t_0) \cdot \int_{t_0}^{t} \bSigma^{-1}(s,t_0) \bA^T\bH(\etab(t),\btheta)\bA  \bSigma^{-T}(s,t_0)  \df s \cdot \bSigma^T(t,t_0) \bF^T(\etab(t),\btheta) \cdot \df t\\
		&=& \left(\bF(\etab(t),\btheta) \bPhi(t) + \bPhi(t) \bF^T(\etab(t),\btheta) + \bA^T\bH(\etab(t),\btheta)\bA \right) \df t,
	\end{eqnarray*}
	which is the result in \eqref{eq:Phi}. 
\end{proof}
\subsection{Relationship between LNA and other methods}
The SDE approach can be viewed as a normal approximation based on a $\tau$-leaping step for the MJP. 
The LNA can be derived either directly from Taylor expansion of  the transition probability of the MJP or the Taylor expansion of the transition density of the SDE. 
The ODE solution can be considered as a limit of the mean trajectory of the MJP when system size $N$ goes to infinity. ODE solution can also be viewed as the deterministic part for SDE \eqref{eq:SDE} and the mean process for LNA based on \eqref{eq:LNAnoise}. 
Figure ~\ref{f:relationshipDiag} depicts relationships between different dynamical system representations as a diagram.
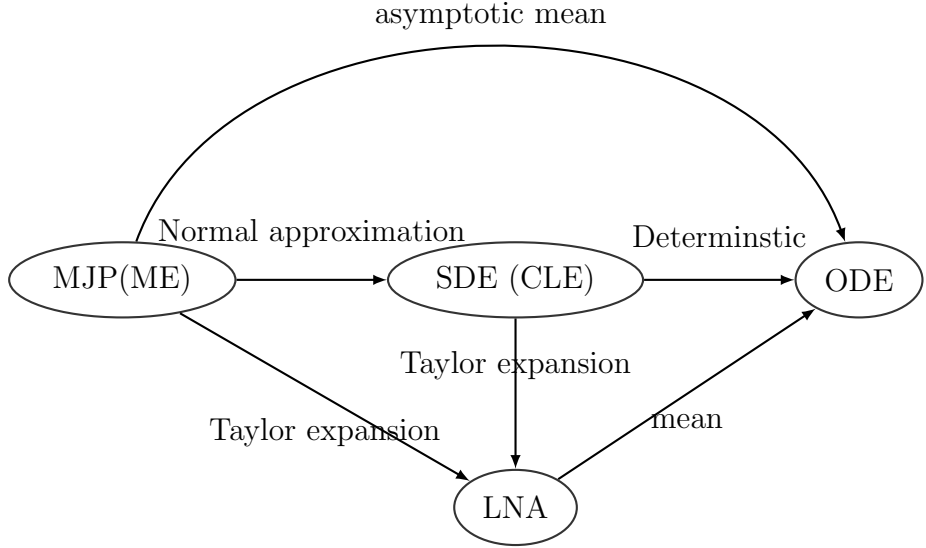
\begin{figure}[h]
	\label{f:relationshipDiag}
	\centering
	\begin{tikzpicture}
	\usetikzlibrary{shapes}
	\tikzstyle{main}=[ellipse, minimum size = 10mm, thick, draw =black!80, node distance = 20mm]
	\tikzstyle{connect}=[-latex, thick]
	\tikzstyle{box}=[rectangle, draw=black!100]
	\node[main] (theta)  {MJP(ME)};
	\node[main] (z) [right=of theta] {SDE (CLE)};
	\node[main] (w) [right=of z] { ODE};
	\node[main] (n) [below=of z] { LNA};
	\path (theta) edge [ connect] node [midway,above=0.7em]  {Normal approximation}(z)
	(z) edge [connect] node [midway,above=0.7em]  {Determinstic}(w)
	(z) edge[connect] node [midway,above=0.1em]  {Taylor expansion}(n)
	(theta) edge[connect] node [midway,below=0.3em]  {Taylor expansion}(n)
	(n) edge[connect] node [midway,below=0.3em]  {mean}(w)
	(theta) edge [ bend left=70,-latex, thick] node [midway,above=0.2em]  {asymptotic mean}(w);
	\end{tikzpicture}
	\caption{The relationship between different dynamical system representations.}
\end{figure}

\section{Non-centered parameterization}
\label{s:noncenter}
In LNA, the latent trajectory $\bX(t)$ is decomposed into the deterministic part $\etab(t)$ plus a stochastic part $\bM(t)$ that follows a multivariate Gaussian distribution with mean $\mathbf{0}$. However, the population size at the $i$-th time interval $\bX_i$ depends on rate parameter $\btheta$ and is correlated with other population sizes $\bX_{j}$s in the trajectory, leading to mixing issues for the MCMC chain, especially when we introduce multiple change points for reproduction number $R_0$. 

Here we take the idea of non-centered parameterization from \citet{papaspiliopoulos2007general,bernardo2003non} and reparameterize the latent trajectrory in terms of residuals $\bX_i - \etab_i$ for $i=1,\dots,T$. 
Given rate parameters $\btheta_{i-1}$, ODE solution $\etab_{0:T}$, fundamental matrix $\bSigma(t_i,t_{i-1})$ and variance matrix $\bPhi_i$ in \eqref{eq:Phi}, the trajectory $\bX_{0:T}$ can be parameterized using standard Gaussian noise $\bxi_{1:T}$ based on the following iterative equations:
\begin{eqnarray}
\label{eq:X0}
\bX_{0} &=& \etab_0, \\
\bX_i &=& \bmu\left(\bX_{i-1} - \etab(t_{i-1}),\Delta t_i, \btheta_{i-1}\right) + \etab_i + \bPhi_i^{1/2}\bxi_i,\\
\label{eq:Xi}&=& \bSigma(t_i,t_{i-1}) \left(\bX_{i-1} - \etab_{i-1} \right) + \etab_i + \bPhi_i^{1/2}\bxi_i, \quad \mbox{for } i = 1,\ldots, T.
\end{eqnarray}
Let $\bM_i :=\bX_i - \etab_i$ denote the residual in grid cell $i$. 
Based on \eqref{eq:Xi}, the residual process satisfies
\begin{eqnarray}
\label{eq:U1}\bM_1 &= &\bPhi_1^{1/2}\bxi_1\\
\label{eq:Ui}\bM_i &=& \bSigma(t_{i-1},t_i)\bM_{i-1} + \bPhi_i^{1/2}\bxi_i, \qquad i = 2,\ldots, T. 
\end{eqnarray}
Then $\bM_{0:T}$ can be viewed as a Gaussian Markov random field with mean $\mathbf{0}$ that follows the Markov property on a chain graph. 
Let $\bSigma_i$ be the abbreviated  notation of $\bSigma(t_i,t_{i-1})$ and $\bP_i = \bPhi_i^{1/2}$. From \eqref{eq:Ui} , $\bM_i$ can be written as
\begin{eqnarray*}
	\bM_{i}
	& = &\bSigma_{i-1}\bM_{i-1} + \bP_{i}\bxi_{i}\\
	& = &\bSigma_{i-1} \left(\bSigma_{i-2}\bM_{i-2} + \bP_{i-1}\bxi_{i-1} \right) + \bP_{i}\bxi_{i}\\
	& = &\bSigma_{i-1}\bSigma_{i-2}\bM_{i-2}+ \bSigma_{i-1}\bP_{i-1}\bxi_{i-1} + \bP_{i}\bxi_{i}\\
	& = &\bSigma_{i-1}\bSigma_{i-2}\cdots\bSigma_1\bP_1\bxi_1 + \cdots + \bP_{i}\bxi_{i}\\
	&  = &\sum_{k=1}^{i}(\prod_{j=k}^{i-1} \bSigma_j) \bP_{k}\bxi_{k}.
\end{eqnarray*}

Since $\bSigma_i$ and $\bP_i$ are governed by rate parameters $\btheta_{i-1}$ and initial value $\bX_0$, then we define the transform matrix $\bL(\bX_0,\btheta_{0:T}) \in \mathbb{R}^{2T\times2T}$, 
\begin{eqnarray}
\bL(\bX_0,\btheta_{0:T}) = \left(\begin{array}{cccccc}
\bP_1 & \mathbf{0} & \mathbf{0}&\cdots & \mathbf{0} & \mathbf{0}\\
\bSigma_1\bP_1 & \bP_2 & \mathbf{0} & \cdots & \mathbf{0} & \mathbf{0}\\
\bSigma_2\bSigma_1\bP_1  & \bSigma_2\bP_2 & \bP_3& \cdots & \mathbf{0} & \mathbf{0}\\
\cdots & \cdots & \cdots & \ddots & \cdots & \cdots\\
\bSigma_{T-2}\cdots \bSigma_1\bP_1 & \bSigma_{T-2}\cdots \bSigma_2\bP_2 & \bSigma_{T-2}\cdots \bSigma_2\bP_3 & \cdots & \bP_{T-1} & \mathbf{0} \\
\bSigma_{T-1}\cdots \bSigma_1\bP_1 & \bSigma_{T-1}\cdots \bSigma_2\bP_2 & \bSigma_{T-1}\cdots \bSigma_3\bP_3 & \cdots & \bSigma_{T-1}\bP_{T-1} & \bP_{T}
\end{array}
\right). 
\end{eqnarray}
A linear relationship between $\bX_{1:T}$ and the reparameterized noise $\bxi_{1:T}$ can be established with the help of the above transform matrix $\bL$, 
\begin{eqnarray}
\label{eq:LNAnoise}
\left(\begin{array}{c}
\bX_1\\
\vdots\\
\bX_T
\end{array}\right) = \left(
\begin{array}{c}
\etab_1\\
\vdots\\
\etab_T
\end{array}
\right) + \bL(\bX_0,\btheta_{0:T})\left(
\begin{array}{c}
\bxi_1\\
\vdots\\
\bxi_T
\end{array}
\right).
\end{eqnarray}

Instead of directly updating $\bX_{1:T}$, we will apply the above transform and update the Gaussian noise $\bxi_{1:T}$ instead. The MCMC approach will focus on sampling parameter  $I_0,R_0,\gamma, \bdelta_{1:T}$, $\bxi_{1:T}, \sigma$ with
the posterior likelihood 
\begin{eqnarray*}
& & \Pr(I_0,R_0,\gamma, \bdelta_{1:T}, \bxi_{1:T}, \sigma|\bg) \\
& \propto&  \Pr(\bg|I_0,R_0,\gamma, \bdelta_{1:T}, \bxi_{1:T}, \sigma)\Pr(I_0)\Pr(R_0)\Pr(\gamma)\Pr(\bdelta_{1:T})\Pr(\bxi_{1:T})\Pr(\sigma)\\
& \propto& \Pr(\bg|\bX_{0:T}, \btheta_{0:T})\Pr(I_0)\Pr(R_0)\Pr(\gamma)\Pr(\bdelta_{1:T})\Pr(\bxi_{1:T})\Pr(\sigma).
\end{eqnarray*}

In summary,the transformation that allows us to move from parameterization in terms of $\bX_{0:T}, \theta_{0:T}$ to the parameterization in terms of $I_0,R_0,\gamma, \bdelta_{1:T}$, $\bxi_{1:T}, \sigma$ are based on the following equations: 
\begin{enumerate}
	\item $R_{0_i}:=R_{0}(t_i) = R_0 \cdot \exp(\prod_{j=1}^{i}\delta_j\sigma)$ --- a function of  $R_0$, $\bdelta_{1:i}$ and $\sigma$.
	\item $\beta_i := \beta(t_i) = \dfrac{N R_0(t_i)}{\gamma}$ --- a function of $R_0$, $\bdelta_{1:i}$, $\sigma$ and $\gamma$. 
	\item $\btheta_i = (\beta_i, \gamma)$ ---  a function of $R_0$, $\bdelta_{1:i}$, $\sigma$ and $\gamma$. 
	\item $\btheta_{0:T}$ --- a function of  $R_0$, $\bdelta_{1:T}$, $\sigma$ and $\gamma$. 
	\item $\bX_0 = (N,I_0)^T$. 
	\item $\bX_{1:T} = \etab_{1:T} + \bL(\bX_0,\btheta_{0:T})\bxi_{1:T}$ --- a function of $R_0$, $\bdelta_{1:T}$, $\sigma$, $\gamma$, $I_0$ and $\bxi_{1:T}$. 
\end{enumerate}
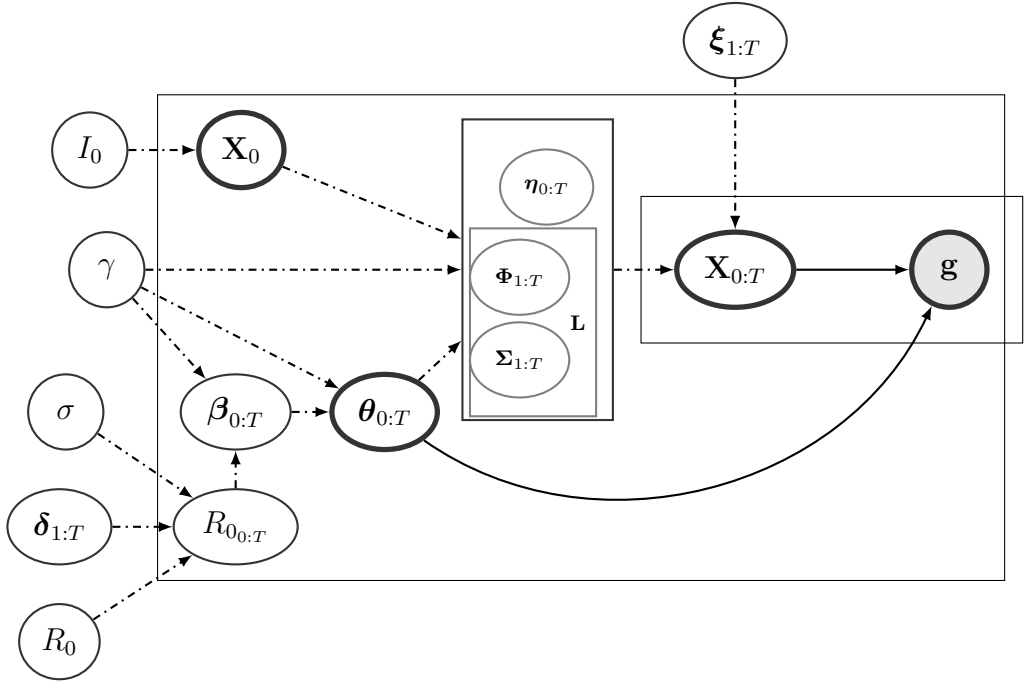
\begin{figure}[h]
	\centering
	\begin{tikzpicture}[title/.style={font=\fontsize{6}{6}\color{black!50}\ttfamily},
	typetag/.style={draw=black!50, font=\scriptsize\ttfamily, anchor=west}]
	\usetikzlibrary{shapes}
	\tikzstyle{main}=[ellipse, minimum size = 10mm, thick, draw =black!80, node distance = 16mm]
	\tikzstyle{main2}=[rectangle, minimum width = 20mm,minimum height = 40mm, thick, draw =black!80, node distance = 16mm]
	\tikzstyle{main3}=[rectangle, minimum width = 15mm,minimum height = 25mm, thick, draw =black!80, node distance = 16mm]
	\tikzstyle{mainnew}=[ellipse, dashed, minimum size = 10mm, thick, draw =black!80, node distance = 16mm]
	\tikzstyle{connect}=[-latex, thick]
	\tikzstyle{depend}=[-latex, thick, dash dot]
	\tikzstyle{box}=[rectangle, draw=black!100]
	\node[main2] (w) [title]{};
	\node[main] (phi)[below=0.1cm of w.west, typetag,xshift = 1mm]  {$\bPhi_{1:T}$};
	\node[main] (SG)[below=1.2cm of w.west, typetag,xshift = 1mm]  {$\bSigma_{1:T}$};
	\node[main3] (L)[below=0.7cm of w.west, typetag,xshift = 1mm]  {$\qquad\qquad\bL$};
	\node[main] (eta)[below=-1.1 cm of w.west, typetag,xshift=0.5cm]  {$\etab_{0:T}$};
	\node[main] (theta)[line width=0.7mm, above left= -0.8 cm and 2.5 cm of w]  {$\bX_0$};
	\node[main] (I0)[left=0.9cm of theta]  {$I_0$};
	\node[main] (z) [line width=0.7mm,below left= -0.5 cm and 0.5cm of w] {$\btheta_{0:T}$};
	\node[main] (delta) [left=0.5cm of z] {$\bbeta_{0:T}$};
	\node[main] (b) [left=1cm of delta] {$\sigma$};
	\node[main] (sigma) [below=0.5cm of delta] {$R_{0_{0:T}}$};
	\node[main] (dtt) [left=0.8cm of sigma] {$\bdelta_{1:T}$};
	\node[main] (sg) [below=0.5cm of dtt] {$R_0$};
	\node[main] (gamma)[left=4.2cm of w]  {$\gamma$};
	\node[main] (v) [line width=0.7mm, right=0.8cm of w] {$\bX_{0:T}$};
	\node[main](xi)[above=2cm of v] {$\bxi_{1:T}$};
	\node[main,fill = black!10] (n) [line width=0.7mm, right= 1.5 cm of v] {$\bg$};
	\path (z) edge [depend] node [midway,above=0.1em] {}(w)
	(w) edge[depend] node [midway,below=0.3em] {}(v)
	(dtt) edge[depend] node  [midway,above=0.1em] {}(sigma)
	(b) edge[depend] node  [midway,above=0.1em] {}(sigma)
	(sg) edge[depend] node  [midway,above=0.1em] {}(sigma)
	(sigma) edge[depend] node  [midway,above=0.1em] {}(delta)
	(gamma) edge[depend] node  [midway,above=0.1em] {}(delta)
	(I0) edge[depend] node  [midway,above=0.1em] {}(theta)
	(gamma) edge[depend] node  [midway,above=0.1em] {}(z)
	(delta) edge[depend] node  [midway,above=0.1em] {}(z)
	(xi) edge[depend] node  [midway,above=0.1em] {}(v)
	(z) edge[bend right=50,connect] node  [midway,above=0.1em] {}(n)
	(gamma) edge[depend] node  [midway,above=0.1em] {}(w)
	(theta) edge [depend] node [midway,above=0.1em] {}(w)
	(v) edge [connect] node [midway,above=0.1em] {}(n);
	\node[rectangle, inner sep=4.4mm,draw=black!100, fit= (n)(v)] {};
	\node[rectangle, inner sep=2mm,draw=black!100, fit= (theta)(sigma)(w)(n)(v)] {};
	\end{tikzpicture}
	\caption{Parameter dependency graph after reparameterization. The root nodes $I_0, \gamma, \sigma, \bdelta_{1:T}, R_0, \sigma$ outside the large box are parameters and latent variables after reparameterization, for which we assign prior distributions. The dash-dotted lines show deterministic relationships and the solid lines show the stochastic dependencies. The grey node denotes the observed data. The figure shows the dependency structure between the transformed parameters and original parameters $\btheta_{0:T}, \bX_{0} \mbox{ and } \bX_{0:T}$.}
	\label{f:dataGenerate2}
\end{figure}

\section{MCMC details}
\label{s:mcapdx}
\subsection{Elliptical slice sampler}
\label{s:ESS}
The elliptical slice sampler, proposed by \citet{murray2010elliptical}, aims at sampling from posterior distributions associated with probability models with a latent \textit{a priori} zero-mean Gaussian random vector $\bX\in \mathbb{R}^d$ with covariance $\Sigma(\btheta)$, i.e., $\bX \sim \mathcal{N}(\mathbf{0},\Sigma(\btheta))$. We use $L(\bY|\bX,\btheta)$ to denote the likelihood function for observed data $\bY$ given latent variable $\bX$ and parameter $\btheta$. Hence, the target posterior distribution for $\bX$ given is $$\Pr(\bX|\bY,\btheta) \propto L(\bY|\bX,\btheta) \mathcal{N}(\bX|\mathbf{0},\Sigma(\btheta))\pi(\btheta),$$
where $\pi(\btheta)$ is the prior distribution for parameter $\btheta$. The goal of elliptical slice sampler is to obtain posterior samples of latent variable $\bX$ from $p(\bX|\bY,\btheta)$. 
The proposal step in elliptical sampling consists of two parts: (1)
proposing an auxiliary random vector  $\bZ \in \mathbb{R}^d$ from distribution $\mathcal{N}\left(\mathbf{0},\Sigma(\btheta)\right)$, (2) proposing a variable $\alpha \in [0,2\pi]$ as an angle parameter. In elliptical slice sampler, a new state $(\bX',\bZ')$ is proposed by rotating the previous state $(\bX,\bZ)$ with angle $\alpha$,
\begin{eqnarray}
\label{eq:invariant}
\bX' = \bX \cos(\alpha) + \bX\sin(\alpha)\\
\bZ' = \bZ \sin(\alpha) - \bZ \cos(\alpha). 
\end{eqnarray}
For any given $\alpha$, this transition leaves the joint prior probability invariant, i.e, 
$$\mathcal{N}\left(\bX|\mathbf{0},\Sigma\right)\mathcal{N}\left(\bZ|\mathbf{0},\Sigma\right) = \mathcal{N}\left(\bX'|\mathbf{0},\Sigma\right)\mathcal{N}\left(\bZ'|\mathbf{0},\Sigma\right).$$

Hence, $(\bX',\bZ')$ are considered at the proposed state and the ratio and the propose transition probability from $(\bX,\bZ)$ to $(\bX',\bZ')$ should equal that from $(\bX',\bZ')$ to $(\bX,\bZ)$, i.e 
$$\dfrac{\Pr\left( (\bX',\bZ')\rightarrow(\bX,\bZ) \right)}{\Pr\left( (\bX,\bZ)\rightarrow(\bX',\bZ') \right)} = 1.$$
The algorithm for elliptical slice sampler is given in Algorithm~\ref{alg:ESS}. 
\begin{algorithm}
	\caption{Elliptical slice sampler for posterior distribution $\pi(\cdot \mid\bY,\btheta)$}\label{alg:ESS}
	\begin{algorithmic}[1]
		\State {\bf Input: }Latent variable from the previous iteration $\bX\in \mathbb{R}^d$.  Observed data $\bY$, previous updated parameter $\bX$.  
		\State{\bf Output} Updated latent variable value $\bX'$
		\State Sample ellipse $\bZ\sim \mathcal{N}(\mathbf{0},\bSigma(\btheta))$
		\State Compute log-likelihood threshold: 
		sample $U \sim \mbox{Uniform}(0,1)$ and let $$\tau \leftarrow \log L(\bY|\bX,\btheta) + \log(U)$$
		\State Sample angle parameter $\alpha \sim \mbox{Uniform}[0,2\pi]$ and $[\alpha_{\textrm{min}},\alpha_{\textrm{max}}]\leftarrow [\alpha - 2\pi, \alpha]$
		\State $\bX'\leftarrow \bX \cdot \cos\alpha + \bZ\cdot \sin\alpha$
		\While {$\log(L(\bY|\bX',\btheta)) < \tau$ }
		\If{$\alpha < 0$}
		\State $\alpha_{\min} \leftarrow \alpha$
		\Else 
		\State $\alpha_{\max} \leftarrow \alpha$
		\EndIf
		\State Sample $\alpha\sim  \mbox{Uniform}\left(\alpha_{\min},\alpha_{\max}\right).$
		\State Make new proposal:
		$$\bX'\leftarrow \bX \cdot \cos\alpha + \bZ\cdot \sin\alpha$$
		\EndWhile
		\State Return $\bX'$.
	\end{algorithmic}
\end{algorithm}
Notice that iterations will stop only when a new sample is accepted. 
Hence, the elliptical slice sampler has acceptance rate $1$, meaning that it will always update the latent random vector $\bX$ at each MCMC iteration. 	

\subsection{MCMC algorithm for the LNA-based SIR model}
In this framework, the observed data are the genealogy $\bg$ estimated from a sample of sequences from virus hosts. 
The sufficient statistics for SIR structured coalescent likelihood would be the coalescent times $\mathcal{T}$ and sampling times $\mathcal{S}$. The unknown parameters and the latent variables are
\begin{enumerate}
	\item The initial number of infected individuals: $I_0$. The initial population is parameterized as $\bX_0 = (N, I_0)$, suppressing that $S_0 = N  - I_0$ and that there are no recovered individuals at time 0.
	\item The initial basic reproduction number $R_0$. 
	\item The removal rate $\gamma$. 
	\item The hyperparameter $\sigma$ that controls the smoothness of $R_0(t)$ trajectory. 
	\item The parameters modeling the first order differences in $\log(R_0(t)): \bdelta_{1:T}$. Note that assuming $\delta_0 = 1$, the infection rate $\beta_i$ can be obtained as
  \begin{eqnarray}
  \beta_i = \dfrac{\gamma}{N} \cdot R_{0} \exp\left(\sum_{k=0}^{i}\sigma \delta_{k}\right).
  \end{eqnarray} The parameter $\btheta_{0:T}$ can be obtained from $R_0, \bdelta_{1:T}, \gamma$ and $\sigma$. 
	\item Random noise for the population trajectory at $t_1,\ldots,t_T$, i.e.  $\bxi_{1:T}$, with  $\bxi_i\sim_{iid} \mathcal{N}(0,\bI)$ \textit{a priori}. The latent SIR trajectories $\bX_{0:T}$ can be recovered from $\btheta_{0:T}$, $\bX_0$ and random noise $\bxi_{1:T}$. 
\end{enumerate}
The MCMC update for parameters and latent variables is given in Algorithm \ref{MCMCstep}. 
\begin{algorithm}[h]
	\caption{Updating rule in the LNA-based MCMC algorithm}\label{MCMCstep}
	\begin{algorithmic}[1]
		\State {\bf Input: }Parameter values from the previous interation $I_0, R_0, \gamma, \bdelta_{1:T}, \sigma, \bxi_{1:T}$, geneology $\bg$. Proposal density $q_1(\cdot|\cdot)$, $q_2(\cdot|\cdot)$ for updating the initial number of infected individuals and the removal rate. 
		\State{\bf Output} Updated parameters values 
		\State Calculate $\bX_{0:T}$, $\btheta_{0:T}$ based on $I_0, R_0, \gamma, \bdelta_{1:T}, \sigma, \bxi_{1:T}$. 
		\State Propose $I'_0$ based on $q_1(\cdot|I_0)$, then $\bX_{0:T}$ will be deterministically updated to $\bX'_{0:T}$ according to $I'_0$, $R_0, \gamma, \bdelta_{1:T}, \sigma, \bxi_{1:T}$. 
		\State Accept $(I'_0, \bX'_{0:T})$ with acceptance probability 
		$$a \leftarrow \min\left(1,\dfrac{\Pr(\bg|\btheta'_{0:T},\bX'_{0:T})\Pr(I'_0)q_1(I_0|I'_0)}{\Pr(\bg|\btheta_{0:T},\bX_{0:T})\Pr(I_0)q_1(I'_0|I_0)}\right).$$
		\State Propose $\gamma'$ based on $q_2(\cdot|\gamma)$, then $\bX_{0:T}, \btheta_{0:T}$ will be deterministically updated to $\bX'_{0:T}, \btheta'_{0:T}$ according to $I_0$, $R_0, \gamma', \bdelta_{1:T}, \sigma, \bxi_{1:T}$. 
		\State Accept $(\gamma', \bX'_{0:T}, \btheta'_{0:T})$ with acceptance probability 
		$$a \leftarrow \min\left(1,\dfrac{\Pr(\bg|\btheta'_{0:T},\bX'_{0:T})\Pr(\gamma')q_2(\gamma|\gamma')}{\Pr(\bg|\btheta_{0:T},\bX_{0:T})\Pr(\gamma)q_2(\gamma'|\gamma)}\right).$$
		\State Let $\bU = (\log(R_0), \bdelta_{1:T}, \log(\sigma))$, then $\bU$ \textit{a priori} follows a multivariate normal distribution. Use elliptical slice sampler to obtain $\bU'$ and get the updated $R'_0, \bdelta'_{1:T}$ and $\sigma'$. $\bX_{0:T}$ will be deterministically updated to $\bX'_{0:T}$ according to $I_0, R'_0, \gamma, \bdelta'_{1:T}, \sigma'$. 
		\State Since  $\bxi_{1:T}$ \textit{a priori} follows a multivariate normal distribution, we use the elliptical slice sampler to obtain $\bxi'_{1:T}$. $\bX_{0:T}$ will be deterministically updated to $\bX'_{0:T}$ according to $I_0, R_0, \gamma, \bdelta_{1:T}, \sigma,\bxi'_{1:T}$. 
	\end{algorithmic}
\end{algorithm}
\subsection{MCMC algorithm for the ODE-based model}
The MCMC algorithm for ODE-based method is similar to the LNA-based MCMC except there is no need to update the Gaussian noise $\bxi_{1:T} $ in the population trajectory. 
The MCMC updates of parameters and latent variables is given in Algorithm~\ref{MCMCODEstep}. 

\begin{algorithm}[h]
	\caption{Updating rule in the ODE-based MCMC algorithm}\label{MCMCODEstep}
	\begin{algorithmic}[1]
		\State {\bf Input: }Parameter values from the previous interation $I_0, R_0, \gamma, \bdelta_{1:T}, \sigma$, geneology $\bg$. Proposal density $q_1(\cdot|\cdot)$, $q_2(\cdot|\cdot)$ for updating the initial number of  infected individuals and the removal rate. 
		\State{\bf Output} Updated parameters values 
		\State Calculate $\bX_{0:T}$, $\btheta_{0:T}$ based on $I_0, R_0, \gamma, \bdelta_{1:T}, \sigma$. 
		\State Propose $I'_0$ based on $q_1(\cdot|I_0)$, then $\bX_{0:T}$ will be deterministically updated to $\bX'_{0:T}$ according to $I'_0$, $R_0, \gamma, \bdelta_{1:T}, \sigma$ based on ODE integration. 
		\State Accept $(I'_0, \bX'_{0:T})$ with acceptance probability 
		$$a \leftarrow \min\left(1,\dfrac{\Pr(\bg|\btheta'_{0:T},\bX'_{0:T})\Pr(I'_0)q_1(I_0|I'_0)}{\Pr(\bg|\btheta_{0:T},\bX_{0:T})\Pr(I_0)q_1(I'_0|I_0)}\right).$$
		\State Propose $\gamma'$ based on $q_2(\cdot|\gamma)$, then $\bX_{0:T}, \btheta_{0:T}$ will be deterministically updated to $\bX'_{0:T}, \btheta'_{0:T}$ according to $I_0$, $R_0, \gamma', \bdelta_{1:T}, \sigma$
		\State Accept $(\gamma', \bX'_{0:T}, \btheta'_{0:T})$ with acceptance probability
		$$a \leftarrow \min\left(1,\dfrac{\Pr(\bg|\btheta'_{0:T},\bX'_{0:T})\Pr(\gamma')q_2(\gamma|\gamma')}{\Pr(\bg|\btheta_{0:T},\bX_{0:T})\Pr(\gamma)q_2(\gamma'|\gamma)}\right).$$
		\State Let $\bU = (\log(R_0), \bdelta_{1:T}, \log(\sigma))$, then $\bU$  \textit{a priori} follows a multivariate normal distribution. Use elliptical slice sampler of obtain $\bU'$ and get the updated $R'_0, \bdelta'_{1:T}$ and $\sigma'$. $\bX_{0:T}$ will be deterministically updated to $\bX'_{0:T}$ according to $I_0, R'_0, \gamma, \bdelta'_{1:T}, \sigma'$ based on ODE integration.
	\end{algorithmic}
	\label{ode-mcmc}
\end{algorithm}

\section{Details of the simulation study}
\subsection{Simulation details for Section~\ref{s:oneReal} of the main text}
\label{s:samptimes}
Here we provide details for the specified sequence/lineage sampling times and number of samples in each simulation scenario: 
\begin{enumerate}
	\item CONST $R_0(t)$: Sampling times: $\mathcal{S} =\{5,10,50,70,80,90\}$, number of samples at each time: $ \left\{2,20,300,300,200,200\right\}$. 
	\item  SD $R_0(t)$: Sampling times: $\mathcal{S} =\{5,10,50,70,80,90\}$, number of samples at each time: $\left\{2,20,200,80,20,20\right\}$. 
	\item NM $R_0(t)$: Sampling times: $\mathcal{S} =\{5,30,50,70,80,90\}$, number of samples at each time: $\left\{2,50,250,100,20,20\right\}$.
\end{enumerate}
\subsection{Simulation details for Section~\ref{s:repeat} of the main text}
\label{s:repdetails}
The $R_0(t)$ trajectory in the simulations is set to
\begin{eqnarray}
R_0(t) =\left\{
\begin{array}{ll}
1.4\times 1.015^{t/2}, t \in [0,30] &  t \in [0,30),\\
1.750 \times 0.975^{t-30} & t \in [30,80],\\
0.494 & t \in (80,90]
\end{array}
\right.
\end{eqnarray}
which is depicted in the left plot of Figure~\ref{f:rep}. 
The initial number of infected individuals is $I_0 = 3$ and the removal rate is set to $\gamma = 0.3$. The population size is fixed to $N = 1,000,000$. Epidemic trajectories are simulated using the SIR Markov jump process (MJP) and are  accepted/rejected based on the following criteria:
\begin{enumerate}
	\item Reject the SIR trajectories that ends before time $90$. The number of infected individuals should never drop to $0$ for $t\in[0,90]$, i.e. $\min_{t \in [0,90]}I(t) > 0$. 
	\item Reject the SIR trajectories with extremely high maximum prevalence: the maximum prevalence should be less or equal than 12,000, i.e., $\max_{t\in[0,90]}I(t) \leq 12000$. 
	\item Reject SIR trajectories with extremely low maximum prevalence. The maximum prevalence should be greater or equal than 600, i.e.,  $\max_{t\in[0,90]}I(t) \geq 600$.  
\end{enumerate}
The 100 simulated SIR prevalence trajectories are shown in the right plot of Figure~\ref{f:rep}. 
We also plot the ODE trajectory under the same parameter setting. 
\begin{figure}
	\centerline{\includegraphics[width=6.5in]{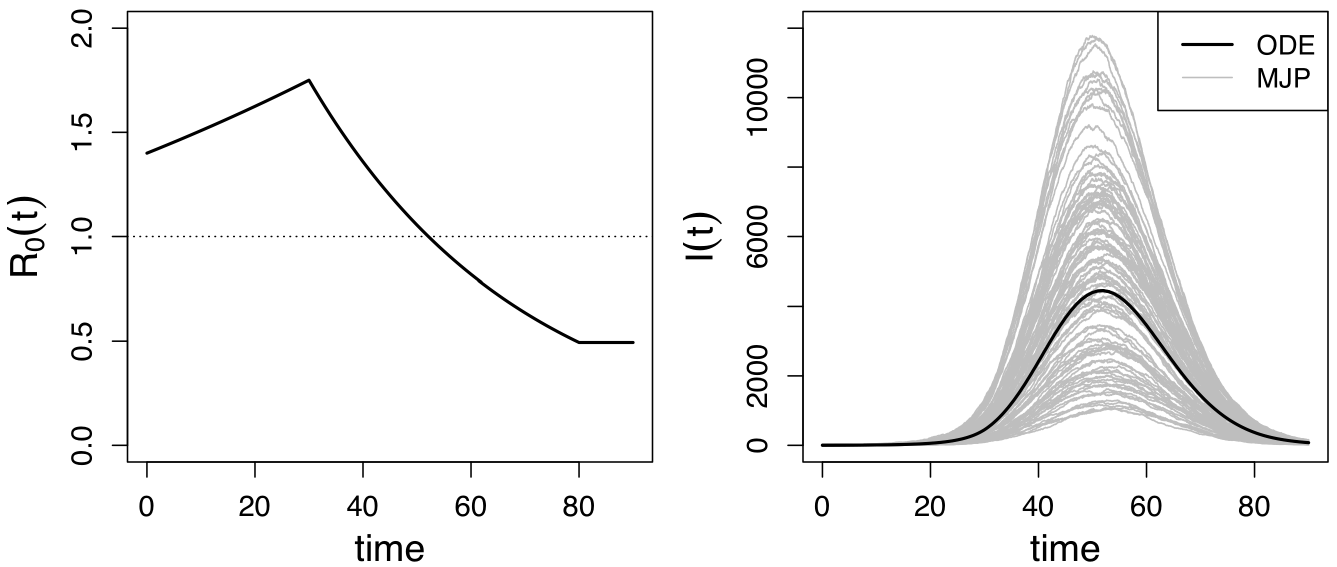}}
	\caption{Repeated simulation setup. Left: $R_0(t)$ trajectory under which the population trajectories are simulated. Right: The 100 simulated prevalence trajectories using MJP and the ODE trajectory under the same parameter setup.}
	\label{f:rep}
\end{figure}



\subsection{Trace plots}
\subsubsection{Trace plots for simulations from Section \ref{s:oneReal} of the main text}
\label{s:traceplotExp}
Figure~\ref{f:exp_trace} shows the trace plots of the log-posterior for the LNA-based method and ODE-based method in the three simulation scenarios  from Section~\ref{s:oneReal}. The effective sample sizes (ESSs) for all parameters range between 100 to 400. 
\begin{figure}
	\centerline{\includegraphics[width=\textwidth]{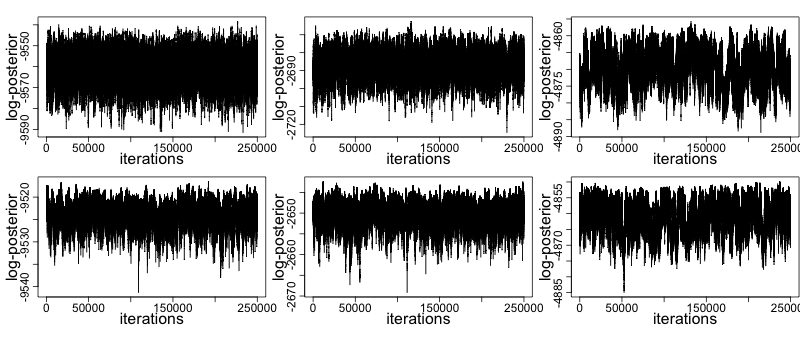}}
	\caption{MCMC trace plots of the log-posterior in the 3 simulation scenarios. 
	Columns correspond to CONST, SD, and NM simulated $R_0(t)$ trajectories. The first row shows the LNA-based results and the second row shows the ODE-based results.}
	\label{f:exp_trace}
\end{figure}
\subsubsection{Trace plots for Ebola data}
\label{s:traceplotEbola}
Figures \ref{f:SL_trace} and \ref{f:SL_ODE_trace} show trace plots of parameters $R_0, I_0, \gamma,\sigma$ for the LNA-based model and ODE-based model respectively applied to the Sierra Leone genealogy. 
Figures \ref{f:Lib_trace} and \ref{f:Lib_traceODE} show the analogous trace plots for the analysis of the Liberia genealogy. We also list posterior medians, 95\% BCIs, and ESSs for each parameter in the MCMC algorithm in Table~\ref{t:MCMCess}.
\begin{figure}
	\centerline{\includegraphics[width=6in]{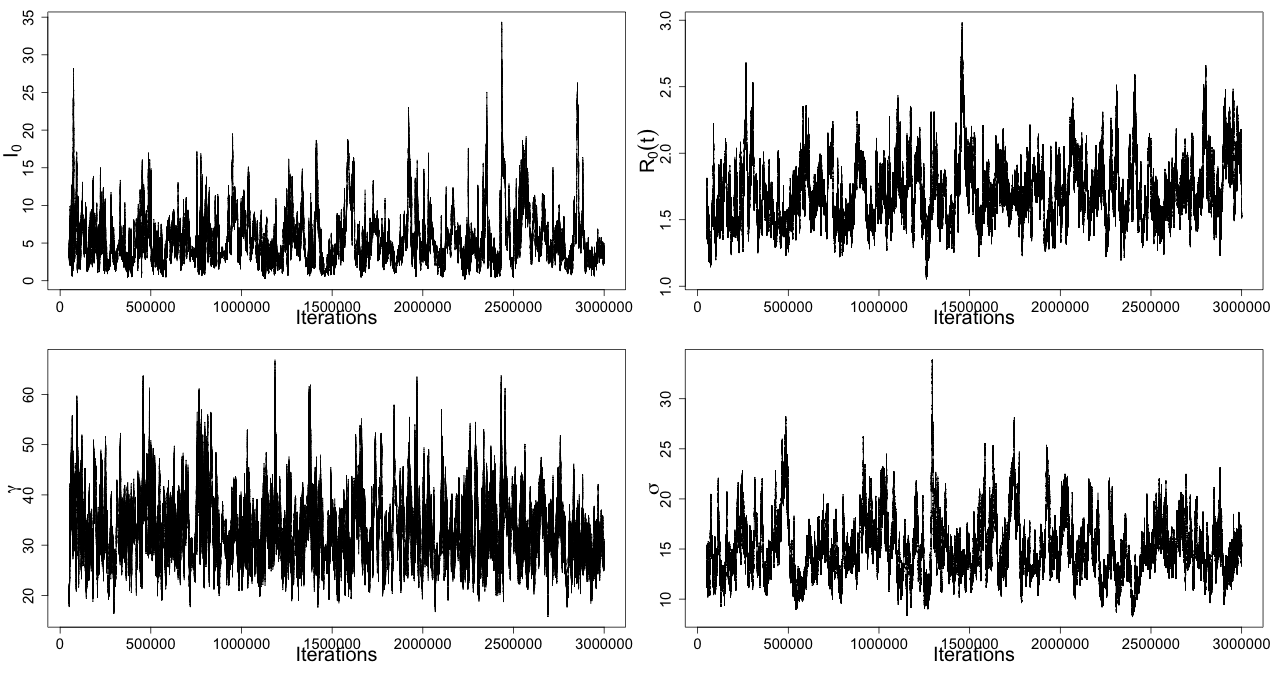}}
	\caption{Trace plots for the LNA-based MCMC algorithm applied to the Ebola genealogy in Sierra Leone. Top right: initial number of infected $I_0$. Top right: initial basic reproduction number $R_0$. Bottom left: removal rate $\gamma$. Bottom right: precision parameter $\sigma$.}
	\label{f:SL_trace}
\end{figure}

\begin{figure}
	\centerline{\includegraphics[width=6in]{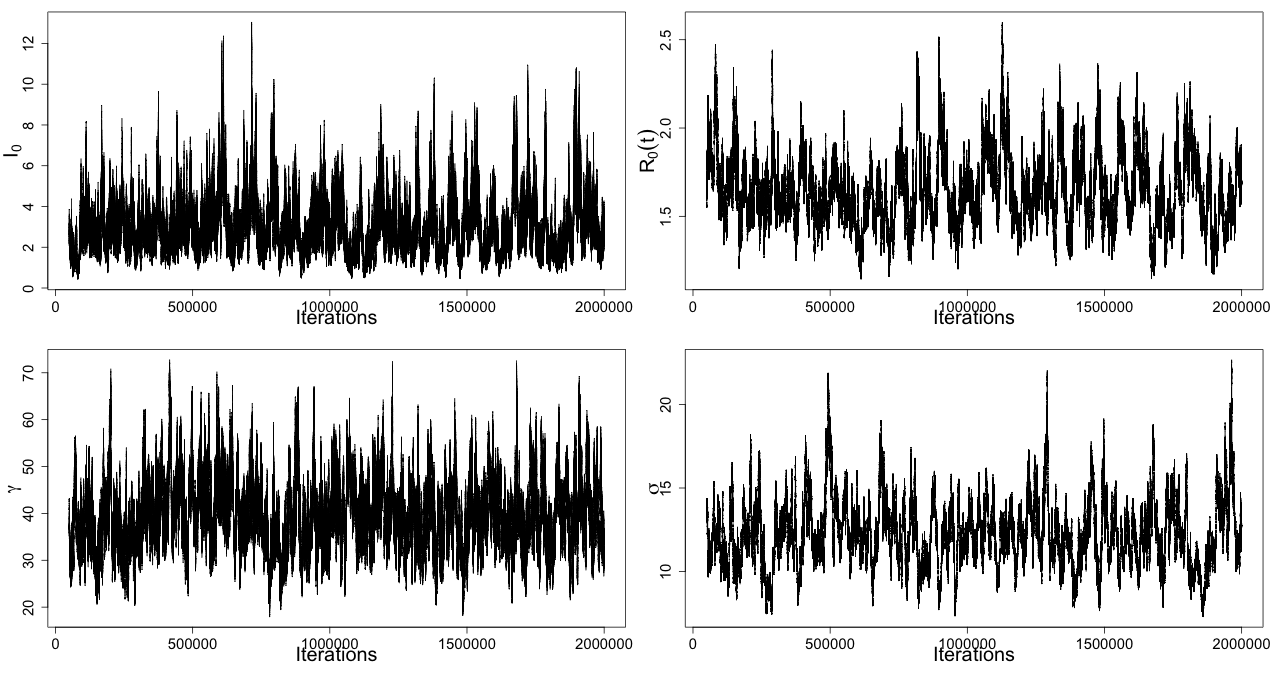}}
	\caption{Trace plots for the ODE-based MCMC algorithm applied to the Ebola genealogy in Sierra Leone. See caption in Figure \ref{f:SL_trace} for the explanation of the plots.}
	\label{f:SL_ODE_trace}
\end{figure}

\begin{figure}
	\centerline{\includegraphics[width=6in]{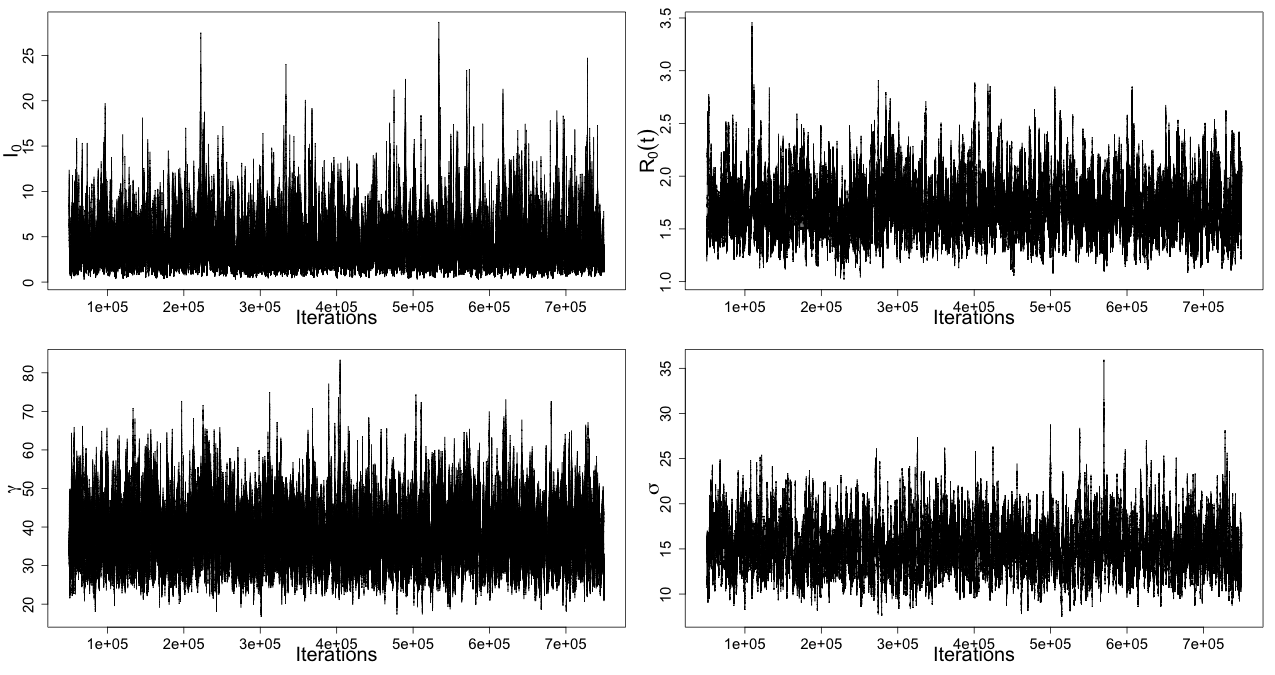}}
	\caption{Trace plots for the LNA-based MCMC algorithm applied to the Ebola genealogy in Liberia. See caption in Figure \ref{f:SL_trace} for the explanation of the plots.}
	\label{f:Lib_trace}
\end{figure}

\begin{figure}
	\centerline{\includegraphics[width=6in]{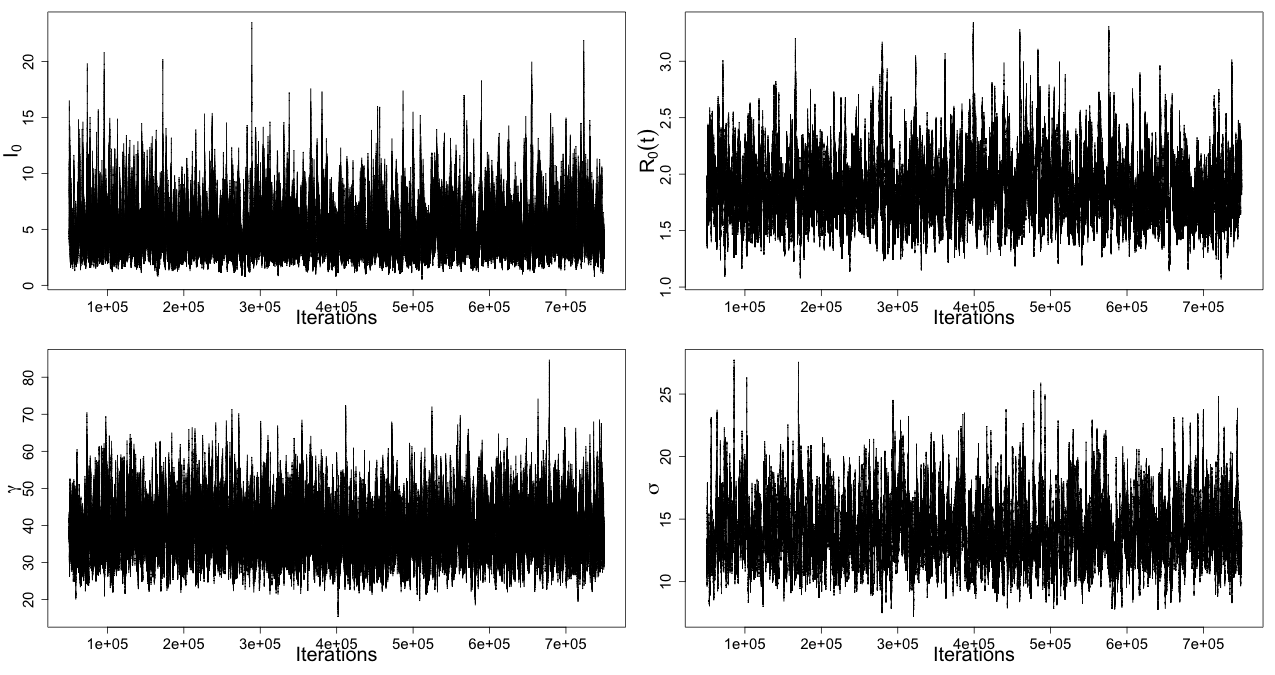}}
	\caption{Trace plots for the ODE-based MCMC algorithm applied to the Ebola data in Liberia. See caption in Figure \ref{f:SL_trace} for the explanation of the plots.}
	\label{f:Lib_traceODE}
\end{figure}
 
\begin{table}[]
	\begin{tabular}{@{}c|cccc|cccc@{}}
		\toprule
		& \multicolumn{4}{c|}{Sierra Leone}                & \multicolumn{4}{c}{Liberia}                       \\
		&          & post med & 95\%BCI           & ESS    &          & post med & 95\%BCI           & ESS     \\ \midrule
		\multirow{4}{*}{LNA} & $I_0$    & 4.63     & {[}1.28,14.41{]}  & 151 & $I_0$    & 3.49     & {[}1.03, 9.95{]}  & 1630 \\
		& $R_0$    & 1.69     & {[}1.33.2.23{]}   & 167 & $R_0$    & 1.67     & {[}1.29,2.24{]}   & 942  \\
		& $\gamma$ & 32.47    & {[}23.08,47.65{]} & 345 & $\gamma$ & 37.21    & {[}25.98,53.13{]} & 1704 \\
		& $\sigma$ & 14.71    & {[}10.33,21.84{]} & 141 & $\sigma$ & 14.83    & {[}10.41,20.70{]} & 870  \\ \midrule
		\multirow{4}{*}{ODE} & $I_0$    & 2.71     & {[}1.10,6.37{]}   & 249 & $I_0$    & 4.31     & {[}1.89,9.27{]}   & 1236 \\
		& $R_0$    & 1.612    & {[}1.30,2.09{]}   & 141 & $R_0$    & 1.83     & {[}1.41.2.44{]}   & 796  \\
		& $\gamma$ & 39.32    & {[}26.63.55.82{]} & 368 & $\gamma$ & 38.31    & {[}27.27.53.43{]} & 1608 \\
		& $\sigma$ & 12.13    & {[}8.61,16.98{]}  & 113 & $\sigma$ & 13.67    & {[}9.78,19.20{]}  & 879  \\ \bottomrule
	\end{tabular}
\caption{Table for posterior median, 95\% BCIs and ESSs for MCMC algorithms applied to Ebola data in Sierra Leone and Liberia.}
\label{t:MCMCess}
\end{table}

\section{Prior sensitivity analysis}
\label{s:sensitivity}
\subsection{Simulations based on single genealogy realizations}
\label{s:sensitivity1}
 In Section~\ref{s:oneReal}, we put informative priors on the removal rate $\gamma$ and explore three different simulation scenarios. 
Although our LNA-based model successfully recovers the $R(t)$ dynamics and SIR trajectories, the posterior density of the removal rate is not too different from its prior in the SD and NM scenarios. 
In this section, we investigate sensitivity of our inferences to changes in the prior of the removal rate $\gamma$. 
For the same genealogies and parameter settings as in Section~\ref{s:oneReal},  we assign weakly informative priors to the removal rate $\gamma$: 
\begin{enumerate}
	\item  CONST $R_0(t)$ scenario: $\gamma\sim \mbox{lognormal}(-1.7,0.25)$,
	\item SD $R_0(t)$ scenario: $\gamma\sim \mbox{lognormal}(-1.7,0.25)$,
	\item NM $R_0(t)$ scenario: $\gamma\sim \mbox{lognormal}(-1.2,0.25)$. 
\end{enumerate} 
For each scenario, we fit a LNA-based model using 300,000 MCMC iterations. 
The first row in Figure~\ref{f:sensitivity} shows the point-wise posterior medians and 95\% BCIs for the basic reproduction number trajectories, $R_0(t)$. 
Our LNA-based method performs well in the CONST and SD scenario. However, for NM scenario, the method fails to fully capture the increase and decrease trend at the beginning and the end of the epidemic. 
The second row in Figure~\ref{f:sensitivity} depicts the prior and posterior densities of the removal rate $\gamma$. 
The LNA-based method estimates the removal rate with good precision in the CONST scenario. 
However, for SD and NM scenario, the removal rate posterior densities are similar to the prior densities, but shift to the right from the truth. 
Posterior summaries of $S(t)$ and $I(t)$ are given in the third and fourth row of Figure~\ref{f:sensitivity}. 
The LNA-based method performs well in recovering the truth in the CONST and SD scenarios. 
In the NM scenario, the true trajectories are still covered by the wide BCIs, but the model seems to underestimate the $S(t)$ and overestimate $I(t)$. 
\begin{figure}
	\centerline{\includegraphics[width=6in]{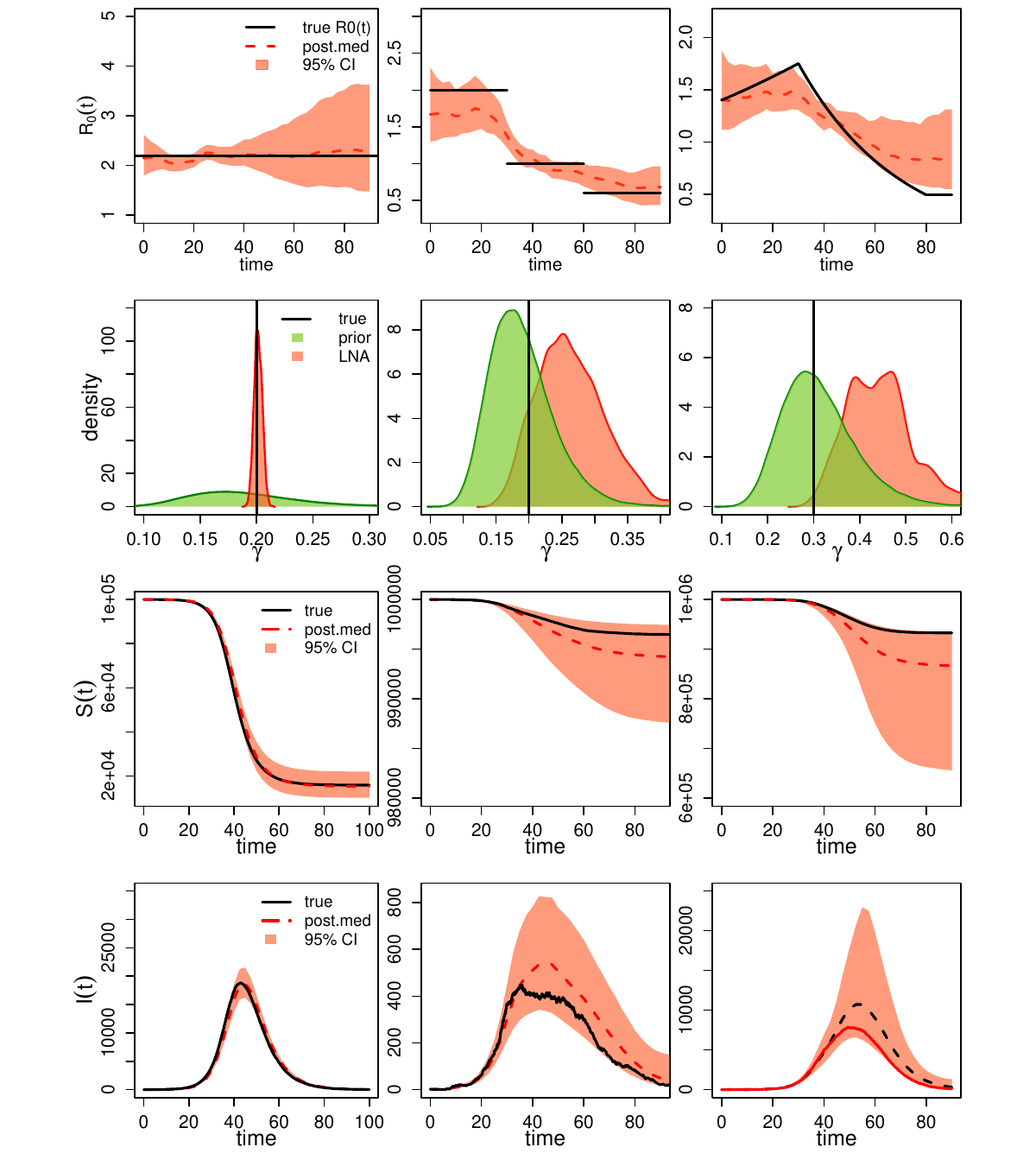}}
	\caption{Analysis of 3 simulation scenarios using the LNA-based method with weakly informative priors. 
	Columns correspond to CONST, SD, and NM simulated $R_0(t)$ trajectories. 
		The first row shows the estimated $R_0(t)$ trajectories for the 3 scenarios, with the black solid lines representing the truth, the red depicting the posterior medians and the red-shaded area showing the 95\% BCIs for the LNA-based method. 
		 The second row corresponds to the estimation of the removal rate $\gamma$. 
		 Posterior density curves from the LNA-base method are shown in red lines compared with prior density curve in green lines. The bottom two 
		rows show the estimated trajectories of $S(t)$ and $I(t)$ respectively.}
	\label{f:sensitivity}
\end{figure}

\subsection{Frequentist properties of posterior summaries}
In this section, we repeat the simulation study in Section~\ref{s:repeat} with a weakly informative prior distribution on recovery: $\gamma \sim \mbox{lognormal}(-1.2,0.25)$.  
We fit LNA-based models to approximate the posterior distribtuion of parameters and latent variables for each genealogy, and compare that with the estimation in Section~\ref{s:repeat} with informative prior on $\gamma$ ($\gamma \sim \mbox{lognormal}(-1.2,0.1)$). 
To evaluate the performance, we use same metric defined in Section~\ref{s:repeat} and generate posterior summary boxplots in Figure~\ref{f:repeat2}. 
\begin{figure}
	\centerline{\includegraphics[width=\textwidth]{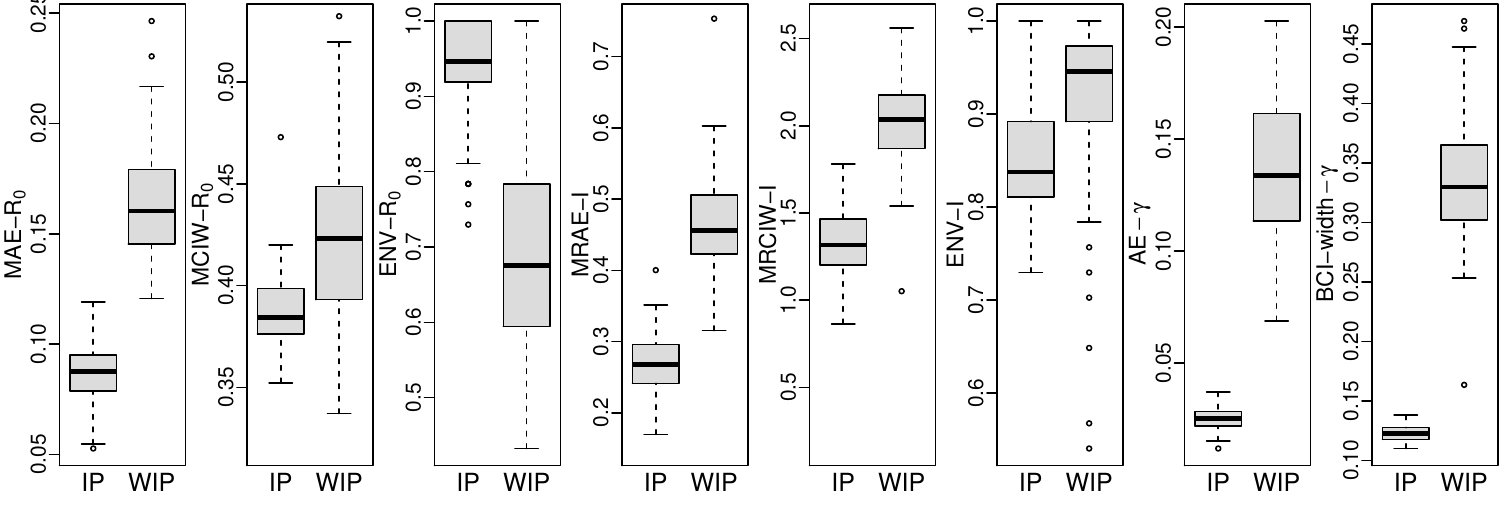}}
	\caption{Boxplots comparing performance LNA-based methods under informative prior (IP) and weakly informative prior (WIP) using 100 simulated genealogies. The first three plots show mean absolute error (MAE), mean credible interval width (MCIW) and envelope for $R_0(t)$ trajectory. The next three plots depict mean relative absolute error (MRAE), mean relative credible interval width (MRCIW), and envelope for $I(t)$ (prevalence) trajectory. The last two plots show the absolute error (AE) and Bayesian credible interval (BCI) width for $\gamma$. }
	\label{f:repeat2}
\end{figure}
Sampling distribution boxplots of $R_0(t)$ posterior summaries are depicted in the left three plots of Figure~\ref{f:repeat2}. 
The LNA-based model with informative prior (IP) on $\gamma$ yields significantly lower MAE and MCIW than that with weakly informative prior (WIP). Compared with IP, the LNA-based model with WIP has really poor envelope for the $R_0(t)$ trajectory. 

Sampling distribution boxplot of $I(t)$ posterior summaries, shown in Figure~\ref{f:repeat2}, are similar to $R_0(t)$ results, with IP yields significantly lower MRAE and lower MRCIW. Somewhat counter intuitively, the WIP cases end up with higher coverage for $I(t)$ trajectory. This is likely caused by the wide BCI under WIP. 

We also report the absolute error (AE) and 95\% BCI widths for the removal rate $\gamma$ in Figure~\ref{f:repeat2}. Though the WIP prior still centered at the truth of $\gamma$, we can see really large absolute error in the removal rate estimation. The 95\% BCI coverage for $\gamma$ under IP is 1. Though WIP yields wider BCIs, the coverage for $\gamma$ is only 0.65.

\end{document}